\documentclass[journal]{IEEEtran}
\ifCLASSINFOpdf
\else
\fi

\usepackage{tikz}
\usetikzlibrary{arrows,positioning,shapes.geometric}
\usetikzlibrary{snakes}
\usetikzlibrary{shapes,arrows}
\usepackage{svg}
\usepackage{comment}
\usepackage[normalem]{ulem}
\usepackage{amssymb}
\usepackage{amsthm}
\usepackage{amsmath}
\usepackage{bbm}
\usepackage{hyperref}
\hypersetup{backref, colorlinks=true, linkcolor=blue, citecolor=blue, urlcolor = blue}
\usepackage{cleveref}
\usepackage{mathtools}
\usepackage{graphicx}
\newtheorem{theorem}{Theorem}
\newtheorem{lemma}{Lemma}
\newtheorem{proposition}{Proposition}
 \newtheorem{corollary}{Corollary}[theorem]
\newtheorem{remark}{Remark}
\newtheorem{definition}{Definition}

\begin{document}
\title{Optimal Causal Rate-Constrained Sampling for a Class of Continuous Markov Processes}

\author{Nian Guo,
        Victoria Kostina% <-this % stops a space
\thanks{N. Guo and V. Kostina are with the Department
of Electrical Engineering, California Institute of Technology, Pasadena, CA, 91125 USA. E-mail: \{nguo,vkostina\}@caltech.edu. This work was supported in part by the National
Science Foundation (NSF) under grants CCF-1751356 and CCF-1956386. A part of this work was presented at the 2020 IEEE International Symposium on Information Theory.}% <-this % stops a space
\thanks{Copyright \copyright 2021 IEEE. Personal use of this material is permitted.  However, permission to use this material for any other purposes must be obtained from the IEEE by sending a request to pubs-permissions@ieee.org.}}% <-this % stops a space
\maketitle

\begin{abstract}
Consider the following communication scenario. An encoder observes a stochastic process and causally decides when and what to transmit about it, under a constraint on the expected number of bits transmitted per second. A decoder uses the received codewords to causally estimate the process in real time. The encoder and the decoder are synchronized in time. For a class of continuous Markov processes satisfying regularity conditions, we find the optimal encoding and decoding policies that minimize the end-to-end estimation mean-square error under the rate constraint. We show that the optimal encoding policy transmits a $1$-bit codeword once the process innovation passes one of two thresholds. The optimal decoder noiselessly recovers the last sample from the 1-bit codewords and codeword-generating time stamps, and uses it to decide the running estimate of the current process, until the next codeword arrives. In particular, we show the optimal causal code for the Ornstein-Uhlenbeck process and calculate its distortion-rate function. Furthermore, we show that the optimal causal code also minimizes the mean-square cost of a continuous-time control system driven by a continuous Markov process and controlled by an additive control signal.
\end{abstract}

\begin{IEEEkeywords}
Causal lossy source coding, sequential estimation, event-triggered sampling, zero-delay coding, rate-distortion theory, control.
\end{IEEEkeywords}

% \fi
\IEEEpeerreviewmaketitle

\section{Introduction}
\subsection{System model and problem setup}
Consider the system in Fig.~\ref{Fig1}. A source outputs a real-valued continuous-time stochastic process $\{X_t\}_{t=0}^T$ with state space $(\mathbb R,\mathcal B_{\mathbb R})$, where $\mathcal B_{\mathbb R}$ is the Borel $\sigma$-algebra on $\mathbb R$.  
\begin{figure}[h!]
\centering
\includegraphics[trim = 35mm 232mm 68mm 49mm, clip, width=10cm]{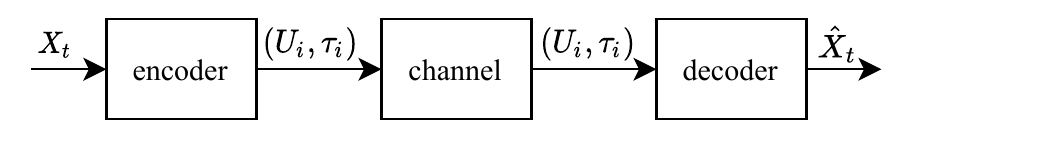}
\centering
\caption{System Model. Sampling time $\tau_i$ and codeword $U_i$ are chosen by the encoder's sampling and compressing policies, respectively.} \label{Fig1}
\end{figure}

An encoder tracks the input process $\{X_t\}_{t=0}^T$ and causally decides to transmit codewords about it at a sequence of stopping times
\begin{equation} \label{timestamp}
0\leq \tau_1\leq \tau_2\leq\dots\leq  \tau_N\leq T
\end{equation} that are decided by a causal sampling policy.  Thus, the total number of time stamps $N$ can be random. The time horizon $T$ can either be finite or infinite.
At time ${\tau}_i$, the encoder generates a codeword $U_i$ according to a causal compressing policy, based on the process stopped at $\tau_i$, $\{X_t\}_{t=0}^{\tau_i}$. Then, the codeword $U_i$ is passed to the decoder without delay through a noiseless channel. At time $t$, $t\in[\tau_i,\tau_{i+1})$, the decoder estimates the input process $X_t$, yielding $\hat X_t$, based on all the received codewords and the codeword-generating time stamps, i.e., $(U_j,\tau_j),~j=1,2,\dots,i$. Note that the encoder and the decoder can leverage the timing information for free due to the clock synchronization and the zero-delay channel.

The communication between the encoder and the decoder is subject to a constraint on the long-term average rate,
\begin{subequations}\label{comm_cons}
\begin{align}\label{comm_cons_a}
    \frac{1}{T}\mathbb E\left[\sum_{i=1}^N\ell(U_i)\right]&\leq R ~\text{(bits per sec)},~(T<\infty),\\  \label{comm_cons_b}
    \limsup_{T\rightarrow\infty} \frac{1}{T}\mathbb E\left[\sum_{i=1}^N\ell(U_i)\right]&\leq R~\text{(bits per sec)}, ~(T=\infty),
\end{align}
\end{subequations}
where $\ell\colon\mathbb Z_+\rightarrow\mathbb Z_+$ denotes the length of its argument in bits, $\ell(x)=\lfloor \log_2(x)\rfloor+1$ for $x>0$, $\ell(0)=1$.
The \emph{distortion} is measured by the long-term average mean-square error (MSE),
\begin{subequations} \label{MSE}
\begin{align}
      \frac{1}{T}\mathbb E \left[\int_{0}^T (X_t-\hat {X}_t)^2 dt\right]&\leq d,~(T<\infty),\\
       \limsup_{T\rightarrow\infty} \frac{1}{T}\mathbb E\left[\int_{0}^T(X_t-\hat X_t)^2dt\right]&\leq d, ~(T=\infty).
\end{align}
\end{subequations}
 We aim to find the encoding and decoding policies that achieve the optimal tradeoff between the communication rate \eqref{comm_cons} and the MSE \eqref{MSE}.

\subsection{The class of processes}\label{IA}
Let $\{\mathcal F_t\}_{t=0}^T$ be the filtration generated by $\{X_t\}_{t= 0}^T$.

Throughout, we impose the following assumptions on the source process.
\begin{itemize}
    \item[(i)]\label{(i)} \textit{(Strong Markov property)} 
$\{X_t\}_{t= 0}^T$ satisfies the strong Markov property:  For all almost surely finite stopping times $\tau\in[0,T]$ and all $t\in[0,T-\tau]$, $X_{t+\tau}$ is conditionally independent of $\mathcal F_{\tau}$ given $X_{\tau}$.

\item[(ii)] \textit{(Continuous paths)} 
$\{X_t\}_{t= 0}^T$ has continuous paths: $X_t$ is almost surely continuous in $t$.

\item[(iii)] \textit{(Mean-square residual error properties)} For all almost surely finite stopping times $\tau\in[0,T]$ and all $t\in[\tau,T]$, 
the mean-square residual error $\tilde X_t=X_t-\mathbb E[X_t|X_{\tau},\tau]$ satisfies:
\begin{itemize}
\item[(iii-a)] $\tilde X_t$ is independent of $\mathcal F_{\tau}$ and $\tilde X_t$ has the Markov property, i.e., for all $r\in[\tau,t]$, $\tilde X_t$ is conditionally independent of $\mathcal F_r$ given $\tilde X_r$.  

\item[(iii-b)] $\tilde X_t$ can be expressed as
\begin{equation}\label{Yt_r}
    \tilde X_t = q_t(s)\tilde X_s+ R_t(s,\tau),
\end{equation}
\end{itemize}where $ s\in [\tau,t]$, $q_t(s)$ is a deterministic function of $(t,s)$, and $R_t(s,\tau)$ is a random process with continuous paths, i.e., $R_t(s,\tau)$ is almost surely continuous in $t$. Furthermore, the random variable $R_t(s,\tau)$ has an even and quasi-concave pdf, and $q_t(t)=1$, $R_t(t,\tau)=0$.
\end{itemize}

We assume that the initial state $X_0=0$ at time $\tau_0=0$ is known both at the encoder and the decoder. For example, any stochastic process of the form $X_t = g_1(t)W_{g_2(t)} + g_3(t)$ satisfies (i)--(iii), where $\{W_t\}_{t\geq 0}$ denotes the Wiener process, $g_1, g_2, g_3$ are continuous deterministic functions of the time $t$, and $g_2$ is positive and non-decreasing in $t$. The parameters in \eqref{Yt_r} for this example process are $q_t(s) = \frac{g_1(t)}{g_1(s)}$ and $R_t(s,\tau) =g_1(t)W_{g_2(t)-g_2(s)}$. The Wiener process, the Ornstein-Uhlenbeck (OU) process, and the continuous L\'{e}vy processes are special cases of this form. These processes are widely used in financial mathematics and physics. There are also other stochastic processes satisfying (i)--(iii), e.g., $X_t = W_{t+c_1}+c_2W_{t}$, where $c_1,c_2\in\mathbb R$, which is expressed by~\eqref{Yt_r} with $q_t(s) = 1$, $R_t(s,\tau) = (1+c_1)W_{t-s}$.

\subsection{Context}
In wireless sensor networks and network control systems of the Internet of Things, nodes are spatially dispersed, communication between nodes is a limited resource, and delays are undesirable. We study the fundamental limits of the communication scenario in which the transmitting node (the encoder) observes a stochastic process, and wants to communicate it in real-time to the receiving node (the decoder).

Related work includes \cite{Imer}--\cite{Zaman}, where it is assumed that the encoder transmits real-valued samples of the input process and that the communication is subject to a sampling frequency constraint or a transmission cost. The causal sampling and estimation policies that achieve the optimal tradeoff between the sampling frequency and the distortion have been studied for the following \emph{discrete-time processes}: the i.i.d process \cite{Imer};  the Gauss-Markov process \cite{Lipsa}; the partially observed Gauss-Markov process \cite{Wu}; and, the first-order autoregressive Markov process $X_{t+1}=aX_t+V_t$ driven by an i.i.d. process $\{V_t\}$ with unimodal and even distribution \cite{chk}\cite{molin}. The first-order autoregressive Markov process considered in \cite{chk}\cite{molin} represents a discrete-time counterpart of the continuous-time process in~\eqref{Yt_r} with $q_t(s) = a^{t-s},R_t(s,\tau)=X_t-a^{t-s}X_s$. Chakravorty and Mahajan \cite{chk} showed that a threshold sampling policy with two constant thresholds and an innovation-based filter jointly minimize a discounted cost function consisting of the MSE and a transmission cost in the infinite time horizon. Molin and Hirche \cite{molin} proposed an iterative algorithm to find the sampling policy that achieves the minimum of a cost function consisting of a linear combination of the MSE and the transmission cost in the finite time horizon, and showed that the algorithm converges to a two-threshold policy.

The optimal sampling policies for some \emph{continuous-time processes} have also been studied: first-order stochastic systems with a Wiener process disturbance \cite{1}; the finite time-horizon Wiener and OU processes \cite{Rabi}; the infinite time-horizon multidimensional Wiener process \cite{Nar}; the infinite-time horizon Wiener process \cite{Sun}; and the OU processes \cite{Zaman} with channel delay. \AA str\"{o}m and Bernhardsson \cite{1} compared uniform and symmetric threshold sampling policies in first-order stochastic systems with a Wiener process disturbance. They showed that the symmetric threshold sampling policy gives a lower distortion than the uniform sampling under the same average
sampling frequency. The optimal causal sampling policies for the Wiener and the OU processes determined in \cite{Rabi}--\cite{Zaman} are threshold sampling policies, whose thresholds are obtained by solving optimal stopping time problems via Snell's envelope. The proofs in \cite{Rabi}--\cite{Zaman} rely on a conjecture about the form of the MMSE decoding policy, implying that the causal sampling policies in \cite{Rabi}--\cite{Zaman} are optimal with respect to the conjectured decoding policy, rather than the optimal decoding policy. Namely, Rabi et al. \cite{Rabi} conjectured that the MMSE decoding policy under the optimal sampling policy is equal to the MMSE decoding policy under deterministic (process-independent) sampling policies without a proof. Nar and Ba\c{s}ar \cite{Nar} arrived at the MMSE decoding policy for the Wiener process by referring to the results in \cite{Nayyar}, where the stochastic processes considered in \cite{Nayyar} are in discrete-time and the increments of the discrete-time process are assumed to have finite support. Yet, the Wiener process is a continuous-time process with Gaussian increments having infinite support. Sun et al. \cite{Sun} and Ornee and Sun \cite{Zaman} assumed that the decoding policy ignores the implied knowledge when no samples are received at the decoder, neglecting the possible influence of the sampling policy on the decoding policy.
Nonparametric estimation of L\'{e}vy processes from uniform non-causal samples has been studied in \cite{Lopez}--\cite{Kappus}.

Although the works \cite{Imer}--\cite{Kappus} did not consider quantization effects, in digital communication systems, real-valued numbers are quantized into bits before a transmission. Quantized event-triggered control schemes have been studied for the following systems: discrete-time linear systems with noise \cite{Khina} and without noise \cite{Yoshi}; continuous-time linear time-invariant (LTI) systems without noise \cite{Pearson}\cite{Kho2019} and with bounded noise \cite{Le2010}--\cite{Talla2016}; partially-observed continuous-time LTI systems without noise \cite{Tw}\cite{Kof} and with bounded noise \cite{Abd}. The quantized event-triggered control schemes in \cite{Khina}--\cite{Abd} are designed to stabilize the systems. The optimality of the proposed schemes was not considered in \cite{Khina}--\cite{Abd}.
In our previous work~\cite{Nian}, we introduced an information-theoretic framework for studying jointly optimal sampling and quantization policies by considering a long-term average bitrate constraint. We showed that the optimal event-triggered sampling policy for the Wiener process remains a two-threshold policy even under a bitrate constraint, while the optimal deterministic (process-independent) sampling policy is uniform.

\subsection{Contribution}
In this paper, we leverage the information-theoretic framework of our prior work \cite{Nian}, introduced in the context of the Wiener process, to study the jointly optimal sampling and quantization policies for the wider class of continuous-time processes introduced in Section~\ref{IA}. Unlike \cite{Nian}, where we rely on previous results \cite{Nar} on the optimal sampling policy for the Wiener process to obtain the optimal causal code, in this paper, we derive the form of the optimal sampling policy from scratch. We prove that the optimal sampling policy is a two-threshold policy whether or not quantization is taken into account. We show that the optimal causal compressor is a sign-of-innovation compressor that generates $1$-bit codewords representing the sign of the process innovation since the last sample. This surprisingly simple structure is a consequence of both the real-time distortion constraint \eqref{MSE}, which penalizes coding delays, and the symmetry of the innovation distribution (iii), which ensures the optimality of the two-threshold sampling policy. Compared to the previous work on sampling of continuous-time processes \cite{Rabi}--\cite{Zaman}, our results apply to a wider class of processes, namely, the processes satisfying (i)--(iii) in Section~\ref{IA}. Furthermore, we confirm the validity of the conjecture on the MMSE decoding policy in \cite{Rabi}--\cite{Nar}. To do so, we use a set of tools that differs from that in \cite{Rabi}--\cite{Nar}: where \cite{Rabi}--\cite{Nar} use Snell's envelope to find the optimal sampling policy under the conjecture on the form of the MMSE decoding policy, we apply majorization theory and real induction to find the jointly optimal sampling and decoding policies.
We show that the optimal causal code for the Ornstein-Uhlenbeck process generates a 1-bit codeword once the process innovation crosses one of the two thresholds, and calculate its distortion-rate function. Next, we show that the SOI code remains optimal in a rate-constrained control scenario with a stochastic plant driven by a process satisfying assumptions (i)--(iii) in Section~\ref{IA}. The SOI code minimizes mean-square cost between the desirable state $0$ and the state of the stochastic plant. In contrast to the event-triggered control schemes in \cite{1} and \cite{Khina}--\cite{Abd}, we introduce a bitrate constraint to the control problems;  we consider a wider class of disturbance signals beyond the Wiener process; we show the optimality of the SOI code in minimizing the state MSE. 

A part of this work was presented at the 2020 IEEE International Symposium on Information Theory \cite{Nian_2}; the conference version does not contain Section~\ref{Control_sec} or any proofs.

\subsection{Paper organization}
In Section~\ref{SecII}, we formulate a causal frequency-constrained sampling problem and show the form of the optimal causal sampling policy. In Section \ref{III}, we formally introduce the causal rate-constrained sampling problem and show the optimal causal code. In Section~\ref{Control_sec}, we prove that the causal code introduced in Section~\ref{III} remains optimal in a rate-constrained control system.
 
\subsection{Notation}
We denote by $\{X_t\}_{t=s}^{r}$ the portion of the stochastic process within the time interval $[s,r]$, and denote by $\{X_t\}_{t>s}^{r}$ the portion of the stochastic process within the time interval $(s,r]$. For a possibly infinite sequence $x=\{x_1,x_2,\dots\}$, we write $x^i=\{x_1,x_2,\dots,x_i\}$ to denote the vector of its first $i$ elements. For a continuous random variable $X$, we denote its pdf by $f_X$. We denote by $\mathrm{Supp}(f_X)\triangleq \{x\colon f_X(x)>0\}$ the support of $f_X$. We use $\sigma(\cdot)$ to denote the $\sigma$-algebra of its argument. We use $X\leftarrow Y$ to represent a substitution of $X$ by $Y$.

\section{Causal frequency-constrained sampling}\label{SecII}
Before we show the optimal causal code in Section~\ref{III}, we formulate the causal frequency-constrained sampling problem and find the optimal tradeoff between the sampling frequency and the MSE. In Theorem~\ref{prop2} in Section~\ref{OSP} below, we find the form of the optimal sampling policy. We will show in Theorem~\ref{thm1} in Section~\ref{iv} that when coupled with an appropriate compressing policy, the optimal causal sampling policy in Theorem~\ref{prop2} attains the optimal tradeoff between the communication rate and the MSE.
\subsection{Causal frequency-constrained code}\label{SecIII}
Allowing the encoder to transmit real-valued samples $U_i=X_{\tau_i}$ instead of the $\mathbb Z_+$-valued codewords $U_i$, and replacing the bitrate constraint \eqref{comm_cons} by the average sampling frequency constraint \begin{subequations}\label{comm_cons_f}
\begin{align}\label{comm_cons_f_a}
    \frac{\mathbb E[N]}{T}&\leq F~(\text{samples per sec}),~(T<\infty),\\\label{comm_cons_f_b}
   \limsup_{T\rightarrow\infty} \frac{\mathbb E[N]}{T}&\leq F~(\text{samples per sec}),~(T=\infty),
\end{align}   
\end{subequations}
where $N$ is the total number of stopping times in \eqref{timestamp}, we obtain the problem of \emph{causal frequency-constrained sampling}.
Next, we formally define causal sampling and decoding policies.
\begin{definition}[($F,d,T$) causal frequency-constrained code]\label{fcode}
A time horizon-$T$ causal frequency-constrained code for the stochastic process $\{X_t\}_{t=0}^T$ is a pair of causal sampling and decoding policies, characterized next. 

\begin{itemize}
    \item[1.] The causal \emph{sampling} policy is a collection of stopping times $\tau_1,\tau_2,\dots$ \eqref{timestamp} adapted to the filtration $\{\mathcal F_t\}_{t=0}^T$ at which samples are generated.

\item[2.]  Given a causal sampling policy, the real-valued samples $\{X_{\tau_j}\}_{j=1}^i$ and sampling time stamps $\tau^i$, the MMSE decoding policy is
\begin{equation}\label{opt_dec_s}
    \bar X_t = \mathbb E[X_t|\{X_{\tau_j}\}_{j=1}^i,\tau^i, t<\tau_{i+1}], ~t\in[\tau_i,\tau_{i+1}).
\end{equation}
\end{itemize}

In an $(F,d,T)$ code, the average sampling frequency must satisfy \eqref{comm_cons_f}, while the MSE must satisfy 
\begin{subequations}\label{MSEf}
\begin{align}
    \frac{1}{T}\mathbb E\left[\int_{0}^T(X_t-\bar{X}_t)^2\right]&\leq d,~(T<\infty),\\\label{MSEf_b}
    \limsup_{T\rightarrow\infty}\frac{1}{T}\mathbb E\left[\int_{0}^T(X_t-\bar{X}_t)^2\right]&\leq d,~(T=\infty).
\end{align}
\end{subequations}
\end{definition}

Allowing more freedom in designing the decoding policy will not lead to a lower MSE, since \eqref{opt_dec_s} is the MMSE estimator. Note that we cannot immediately simplify the expectation in \eqref{opt_dec_s} using the strong Markov property of $\{X_t\}_{t= 0}^T$ ((i) in Section~\ref{IA}) at this point, since the expectation is also conditioned on $t<\tau_{i+1}$. We will show in Corollary~\ref{cor1} below that under the optimal causal sampling policy, \eqref{opt_dec_s} can indeed be simplified to \eqref{cor17}.

In this work, we focus on causal sampling policies satisfying the following assumptions.
\begin{enumerate}
    \item[(iv)] The sampling interval between any two consecutive stopping times, $\tau_{i+1}-\tau_i$, satisfies 
\begin{equation}\label{si}
    \mathbb E[\tau_{i+1}-\tau_i]<\infty,~i=0,1,\dots,
\end{equation}
and the MSE within each interval satisfies
\begin{equation}\label{si_d}
    \mathbb E\left[\int_{\tau_i}^{\tau_{i+1}}(X_t-\bar X_t)^2dt\right]<\infty, ~i=0,1,\dots
\end{equation}
\item[(v)] The Markov chain $\tau_{i+1} - \tau_i - \{X_t\}_{t=0}^{\tau_i}$ holds for all $i=0,1,\dots$
\item[(vi)] For all $i=0,1,\dots$, the conditional pdfs $f_{\tau_{i+1}|\tau_i}$ exist.
\end{enumerate}
Note that \eqref{si} holds trivially if $T<\infty$. Sun et al. \cite{Sun} and Ornee and Sun \cite{Zaman} also assumed \eqref{si} in their analyses of the infinite time horizon problems for the Wiener \cite{Sun} and the OU \cite{Zaman} processes. We use \eqref{si_d} to obtain a simplified form of the distortion-frequency tradeoff for time-homogeneous processes (see \eqref{solvet} below). Furthermore, \eqref{si_d} allows us to prove that the optimal sampling intervals $\tau_{i+1}-\tau_i$ form an i.i.d. process (see \eqref{stp} below). We use (v), (vi) to show that the optimal sampling policy is a symmetric threshold sampling policy in the frequency-constrained setting. See Appendix \ref{pf_rmk_plus} for a sufficient condition on the stochastic process for the optimal sampling policy to satisfy (v). For example, in the infinite time horizon, stochastic processes of the form $X_t=cW_{at}+bt$ satisfy the sufficient condition. Assumption (v) implies that the stopping times form a Markov chain. In contrast, the sampling intervals of causal sampling policies are assumed to form a regenerative process in \cite{Sun}\cite{Zaman}.

To quantify the tradeoffs between the sampling frequency~\eqref{comm_cons_f} and the MSE \eqref{MSEf}, we introduce the distortion-frequency function.

\begin{definition}[Distortion-frequency function (DFF)]\label{Dsf} The DFF for causal frequency-constrained sampling of the process $\{X_t\}_{t=0}^T$ is the minimum MSE achievable by causal frequency-constrained codes,
\begin{equation}\label{dsf_eq}
\begin{aligned}
    \underline{D}(F) \triangleq &\inf\{d:\exists~(F,d,T)~\text{causal}\\ &\text{frequency-constrained code satisfying (iv), (v), (vi)}\}.
\end{aligned}
\end{equation}
\end{definition}

In the causal frequency-constrained sampling scenario, we say that a causal sampling policy is \emph{optimal} if, when succeeded by the MMSE decoding policy \eqref{opt_dec_s}, it forms an $(F,d,T)$ code with $d=\underline{D}(F)$.

\subsection{Optimal causal sampling policy}\label{OSP}
 In Theorem~\ref{prop2} below, we show that the optimal sampling policy is a two-threshold policy that is symmetric with respect to the expected value of the process given the last sample and the last  sampling time, henceforth referred to as a \emph{symmetric threshold policy}. In Theorem~\ref{prop3}, we show a simplified form of the  policy for time-homogeneous processes.

\begin{theorem}\label{prop2}
The optimal causal sampling policy in either finite or infinite time horizon for a class of continuous Markov processes satisfying assumptions (i)--(iii) in Section~\ref{IA} is a symmetric threshold sampling policy of the form
\begin{equation}\label{dstp}
\begin{aligned}
\tau_{i+1}=\inf\{t\geq \tau_i: &X_t-\mathbb E[X_t|X_{\tau_i},\tau_i] \\&\notin (-a_i(t,\tau_i),a_i(t,\tau_i))\},
\end{aligned}
\end{equation}
where the threshold $a_i$ is a non-negative deterministic function of $(t,\tau_i)$.
\end{theorem}
\begin{proof}
Appendix~\ref{pfprop2}.
\end{proof}
Theorem~\ref{prop2} shows that the optimal sampling policy is found within a much smaller set of sampling policies than that allowed in Definition~\ref{Dsf}: the input stochastic process $\{X_t\}_{t=0}^T$ is sampled only if the process innovation passes one of two symmetric thresholds. The thresholds depend on $\{X_t\}_{t=0}^T$ only through the current time $t$, the last sampling time, and the number of samples taken until $t$.
Using the form of the sampling policy~\eqref{dstp}, we show that the MMSE decoding policy \eqref{opt_dec_s} simplifies as follows.
\begin{corollary}\label{cor1}
In the setting of Theorem~\ref{prop2}, under the optimal sampling policy \eqref{dstp}, the MMSE decoding policy reduces to
\begin{equation}\label{cor17}
    \bar X_t = \mathbb E[X_t|X_{\tau_i},\tau_i],~t\in[\tau_i,\tau_{i+1}).
\end{equation}
\end{corollary}
\begin{proof}
Appendix~\ref{pfcor1}.
\end{proof}
In the frequency-constrained setting, the expectation in \eqref{cor17} can be calculated at the decoder even without the knowledge of the sampling policy, whereas the expectation in \eqref{opt_dec_s}  depends on the sampling policy at the encoder through the conditioning on the event that the next sample has not been taken yet, i.e., $t<\tau_{i+1}$. Corollary~\ref{cor1} confirms the conjecture in \cite[Eq.(3)]{Rabi} and \cite[Eq.(5)]{Nar} on the form of the MMSE decoding policy. 

\begin{corollary}\label{cor2}
In the setting of Theorem~\ref{prop2}, the optimal causal sampling policy satisfies \eqref{comm_cons_f} with equality.
\end{corollary}
\begin{proof}
Appendix~\ref{pfcor2}.
\end{proof}
Corollary~\ref{cor2} indicates that the inequality in the sampling frequency constraint \eqref{comm_cons_f} can be simplified to an equality.

\begin{corollary}\label{cor3}
In the setting of Theorem~\ref{prop2}, the threshold in \eqref{dstp} satisfies
\begin{align}\label{greater}
    \lim_{\delta\rightarrow 0^+} a_i(t+\delta,\tau_i) \geq a_i(t,\tau_i),~\forall t\in[\tau_i,\tau_{i+1}), i=0,1,\dots
\end{align}
\end{corollary}
\begin{proof}
Appendix~\ref{pfcor3}.
\end{proof}
Corollary~\ref{cor3} implies that the threshold $a_i(t,\tau_i)$, at time $t\in[\tau_i,\tau_{i+1})$, is either right-continuous or has a jump to a larger value. Thus, the continuous-path process $X_t-\mathbb E[X_t|X_{\tau_i},\tau_i]$ in \eqref{dstp} must hit one of the symmetric thresholds $\pm a_i(\tau_{i+1},\tau_i)$ at $t=\tau_{i+1}$.

\begin{definition}[time-homogeneous process]\label{def_th}
The process $\{X_t\}_{t= 0}^T$ is called \emph{time-homogeneous}, if for a stopping time $\tau\in[0,T]$ and a constant $s\in[0,T-\tau]$, $X_{s+\tau}-\mathbb E[X_{s+\tau}|X_{\tau}]$ follows a distribution that only depends on $s$.
\end{definition}

\begin{theorem}\label{prop3}
In the infinite time horizon, the optimal causal sampling policy for time-homogeneous continuous Markov processes satisfying assumptions (i)--(iii) in Section~\ref{IA} is a symmetric threshold sampling policy of the form
\begin{equation}\label{stp}
\begin{aligned}
    \tau_{i+1} = \inf\{t\geq \tau_i: &X_t-\mathbb E[X_t|X_{\tau_i},\tau_i]\\&\notin (-a(t-\tau_i),a(t-\tau_i))\},
\end{aligned}
\end{equation}
where the threshold $a$ is a non-negative deterministic function of $t-\tau_i$.
The optimal threshold of \eqref{stp} is the solution to the following optimization problem,
\begin{equation} \label{solvet}
    \underline{D}(F) = \min_{\substack{\{a(t)\}_{t\geq 0}\colon\\ \mathbb E[\tau_1]=\frac{1}{F}}} \frac{\mathbb E\left[\int_{0}^{\tau_1} (X_t-\mathbb E[X_t]^2) dt\right]}{\mathbb E[\tau_1]}.
\end{equation}
\end{theorem}
\begin{proof}
Appendix~\ref{pfprop3}.
\end{proof}
\begin{remark}\label{rmk1}
In the setting of Theorem~\ref{prop3}, the sampling intervals $\tau_{i+1}-\tau_i$, $i=0,1,\dots$ under a symmetric threshold sampling policy of the form \eqref{stp} are i.i.d.
\end{remark}
Theorem~\ref{prop3} shows that the optimal sampling policy in Theorem~\ref{prop2} can be further simplified for time-homogeneous processes in the infinite time horizon. As a consequence of time homogeneity, thresholds in \eqref{stp} only depend on the time elapsed since the last sampling time. In contrast, the thresholds in \eqref{dstp} depend on the last sampling time as well.

For example, applying \eqref{solvet} to the time-homogeneous process $X_t = cW_{at} + bt$, $a,b,c\in\mathbb R$, $a>0$, we conclude that the sampling threshold that achieves \eqref{solvet} is equal to $a(t) = c\sqrt{\frac{a}{F}}$ and that $\underline{D}(F) = \frac{ac^2}{6F}$.

\section{Causal rate-constrained sampling}\label{III}
In this section, we formally introduce the causal rate-constrained sampling problem, and we leverage Theorem~\ref{prop2} in Section~\ref{OSP} to find the causal code that achieves the optimal tradeoff between the communication rate and the MSE.
\subsection{Causal rate-constrained code}
We formally define encoding and decoding policies, and define a distortion-rate function (DRF) to describe the tradeoffs between \eqref{comm_cons} and \eqref{MSE}.

\begin{definition}[($R,d,T$)  causal rate-constrained codes]\label{rdt} A time horizon-$T$ causal rate-constrained code for the stochastic process $\{X_t\}_{t=0}^T$ is a pair of encoding and decoding policies. The encoding policy consists of a causal sampling policy and a causal compressing policy.
\begin{itemize}
    \item[1.] The causal sampling policy, defined in Definition~\ref{fcode}-1, decides the stopping times \eqref{timestamp} at which codewords are generated.

\item[2.]  The causal compressing policy, characterized by the $\mathbb Z_+$-valued process $\{f_t\}_{t=0}^T$ adapted to $\{\mathcal F_t\}_{t=0}^T$, 
decides the codeword to transmit at time $\tau_i$,
\begin{equation}
   U_i = f_{\tau_i}.
\end{equation}
\end{itemize}

Given an encoding policy, the MMSE \emph{decoding} policy uses the received codewords and codeword-generating time stamps to estimate the process,
\begin{equation}\label{opt_dec}
    \hat X_t = \mathbb E[X_t|U^i,\tau^i,t<\tau_{i+1}], ~ t\in[\tau_i,\tau_{i+1}). 
\end{equation}
In an $(R,d,T)$ code, the lengths of the codewords must satisfy the average communication rate constraint $R$ bits per sec in~\eqref{comm_cons}, while the MSE must satisfy \eqref{MSE}.

\end{definition}
Allowing more freedom in designing the decoding policy will not lead to a lower MSE, because \eqref{opt_dec} is the MMSE estimator. 

\begin{definition}[Distortion-rate function (DRF)]\label{def2}
The DRF for causal rate-constrained sampling of the process $\{X_t\}_{t=0}^T$ is the minimum MSE achievable by causal rate-$R$ codes:
\begin{equation}\label{DOP}
\begin{aligned}
    D(R) \triangleq\inf\{d:~ &\exists~ (R,d,T)~\text{causal rate-constrained}\\ &\text{ code satisfying (iv), (v), (vi)}\}.
\end{aligned}
\end{equation}
\end{definition}
We say that a causal $(R,d,T)$ code is \emph{optimal} if $d=D(R)$.

\subsection{Optimal causal codes}\label{iv}
We proceed to show that the sampling policies in Theorem~\ref{prop2} remain optimal in the scenario of rate-constrained sampling.
Towards that end, we introduce a class of causal codes, namely, the sign-of-innovation (SOI) codes. We prove that an SOI code is the optimal code as long as the process satisfies the assumptions (i)--(iii) in Section~\ref{IA}.

\begin{definition}[A Sign-of-innovation (SOI) code]\label{def}
The SOI code for a continuous-path process $\{X_t\}_{t=0}^T$ consists of an encoding and a decoding policy. Given a symmetric threshold sampling policy in \eqref{dstp} that satisfies (iv)--(vi), at each stopping time $\tau_i$, $i=1,2,\dots$, the SOI encoding policy generates a $1$-bit codeword
\begin{equation} \label{SOI}
U_i=
    \begin{cases} 
      1 & \text{if}\quad X_{\tau_{i}}-\mathbb E[X_{\tau_{i}}|X_{\tau_{i-1}},\tau_{i-1}]=a_{i-1}(\tau_{i},\tau_{i-1}) \\
      0 & \text{if}\quad X_{\tau_{i}}-\mathbb E[X_{\tau_{i}}|X_{\tau_{i-1}},\tau_{i-1}]=-a_{i-1}(\tau_{i},\tau_{i-1}).
   \end{cases}
\end{equation}
At time $\tau_{i}$, the MMSE decoding policy noiselessly recovers $X_{\tau_{i}}$, $i=1,2,\dots$ via the received codewords $U^i$, 
\begin{equation}\label{rc}
    X_{\tau_{i}} = (2U_i-1)a_{i-1}(\tau_{i},\tau_{i-1})+\mathbb E[X_{\tau_{i}}|X_{\tau_{i-1}},\tau_{i-1}], 
\end{equation}
and uses \eqref{cor17} as the estimate of $X_t$ until $U_{i+1}$ arrives.
\end{definition}

\begin{theorem}\label{thm1}
In either finite or infinite time horizon, for a process $\{X_t\}_{t=0}^T$ satisfying assumptions (i)--(iii) in Section~\ref{IA}, the SOI code, whose stopping times are decided by the optimal symmetric threshold sampling policy \eqref{dstp} with average sampling frequency \eqref{comm_cons_f} $F=R$, is the optimal causal code.
\end{theorem}
\begin{proof}
In Appendix~\ref{pfthm1}, we show the converse
\begin{align}\label{converse}
    D(R)\geq \underline{D}(R).
\end{align}
We proceed to show that the equality in \eqref{converse} is achievable by the SOI code. Corollary~\ref{cor3} implies that the $1$-bit codeword in~\eqref{SOI} together with the recovered samples $\{X_{\tau_j}\}_{j=1}^{i-1}$ suffices to recover $X_{\tau_{i}}, i=1,2,\dots$ noiselessly at the decoder. Moreover, since $\ell(U_i)=1$ under a 1-bit SOI compressor, the rate constraint \eqref{comm_cons} is equal the frequency constraint \eqref{comm_cons_f}, i.e., $\mathbb E\left[\sum_{i=1}^N\ell(U_i)\right] = \mathbb E[N]$. Thus, \eqref{converse} is achieved with equality under the SOI code.
\end{proof}

Theorem~\ref{thm1} illuminates the working principle of the optimal causal code for the stochastic processes considered in Section~\ref{IA}: The encoder transmits a $1$-bit codeword representing the sign of the process innovation as soon as the innovation crosses one of the two symmetric thresholds. To achieve the DRF \eqref{DOP}, the optimal causal code uses the minimum compression rate (1 bit per codeword) in exchange for the maximum average sampling frequency $R$.

Theorem~\ref{thm1} shows that the optimal codeword-generating times are the sampling times of the optimal causal sampling policy. Furthermore, the optimal decoding policy only depends on the thresholds of the sampling policy and the sampling time stamps. Thus, finding the optimal causal code is simplified to finding the optimal causal sampling policy.

\subsection{Rate-constrained sampling of the OU process}\label{V}
Using Theorem~\ref{thm1} and \eqref{solvet}, we can easily find the optimal causal code and its corresponding DRF for the OU process by finding the thresholds of the optimal causal sampling policy. The OU process is the solution to the following SDE:
\begin{equation}
    \mathrm{d}X_t=\theta(\mu- X_t)\mathrm{d}t+\sigma \mathrm{d}W_t,
\end{equation}
where $\mu,\theta,\sigma$ are positive constants, and $W_t$ is the Wiener process. The OU process satisfies the conditions in Section~\ref{IA}.  Under the assumption (iv) in Section~\ref{SecIII} and the assumption that the sampling intervals form a regenerative process, Ornee and Sun \cite{Zaman} found the optimal sampling policy for the OU process in the infinite horizon by forming an optimal stopping problem. They solved the optimal stopping problem via the Snell's envelope which requires solving an SDE. We provide an easier method to find the optimal sampling policy for the OU process in Appendix~\ref{pfprop4}. We also show via Theorem~\ref{thm1} that the policy remains optimal when bitrate constraints are present.

Denote
\begin{equation}\label{27a}
R_1(v) \triangleq \frac{v}{\sigma^2} {}_2F_2\left(1,1;\frac{3}{2},2;\frac{\theta}{\sigma^2}v\right),
\end{equation}
\begin{equation}\label{27b}
R_2(v) \triangleq -\frac{v}{2\theta}+\frac{\sigma^2}{2\theta}R_1(v),
\end{equation}
where $_2F_2$ is a generalized hypergeometric function.

\begin{proposition}\label{prop4} For causal coding of the Ornstein-Uhlenbeck process, the optimal causal sampling policy is the symmetric threshold sampling policy given by
\begin{equation}\label{opt_ou}
    \tau_{i+1}=\inf\left\{t\geq \tau_i:|X_t-\mathbb E[X_t|X_{\tau_i},\tau_i]|\geq \sqrt{R_1^{-1}\left(\frac{1}{R}\right)} \right\},
\end{equation}
The DRF under the corresponding SOI code is given by
\begin{equation}\label{dr}
    D(R) = R\cdot R_2\left(R_1^{-1}\left(\frac{1}{R}\right)\right).
\end{equation}

\end{proposition}
\begin{proof}
Appendix~\ref{pfprop4}.
\end{proof}
\section{Rate-constrained control}\label{Control_sec}
The SOI coding scheme introduced in Definition~\ref{def} also applies to the following rate-constrained control scenario.

\begin{figure}[h!]
\centering
\includegraphics[trim = 35mm 220mm 70mm 18mm, clip, width=9cm]{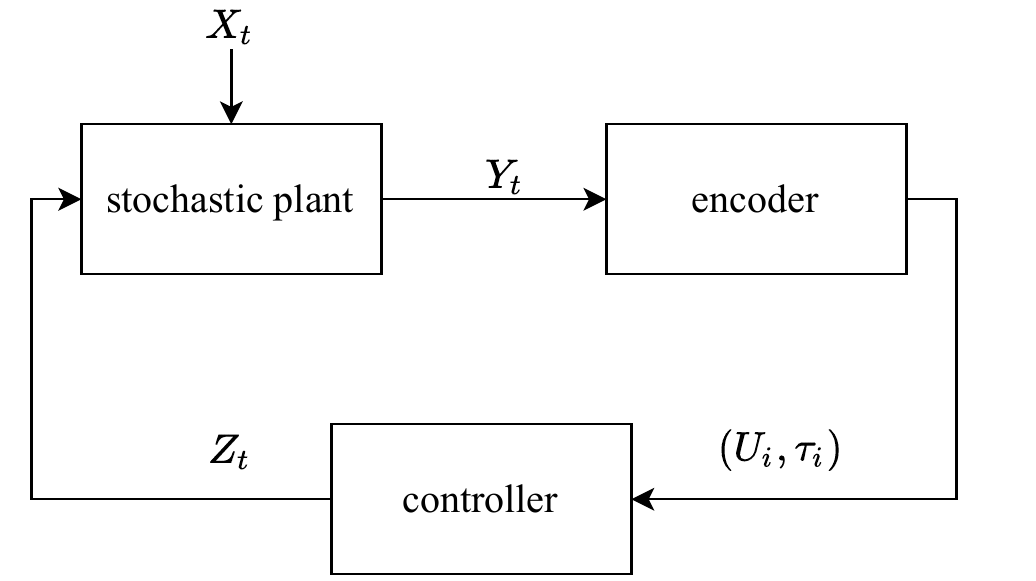}
\centering
\caption{Control system.} \label{Fig_cs}
\end{figure}

The stochastic plant evolves according to
\begin{align}\label{plant}
Y_t=X_t+Z_t,
\end{align}
where $X_t$ is a stochastic disturbance satisfying the assumptions (i)--(iii) in Section~\ref{IA}, and $Z_t$ is the additive control signal output from the controller. The encoder observes $Y_t$, causally decides the stopping times $\tau_1,\tau_2,\dots$ adapted to the filtration generated by $\{Y_t\}_{t=0}^T$, and generates a codeword $U_i$ at each stopping time $\tau_i$ based on its past observations $\{Y_t\}_{t=0}^{\tau_i}$. The controller collects the received codewords to causally form the control signal $Z_t$, with the goal to minimize the mean-square cost on $Y_t$ deviating from the target state $0$,
\begin{align}\label{plant_MSE}
    \frac{1}{T}\mathbb E\left[\int_{0}^T Y_t^2dt\right].
\end{align}
We aim to find the encoding policy satisfying (iv)--(vi) and the control policy that jointly minimize the mean-square cost~\eqref{plant_MSE} under the communication rate constraint \eqref{comm_cons} between the encoder and the controller. 
\begin{proposition}\label{prop_contr}
In the rate-constrained control system, the optimal encoding policy that minimizes the mean-square cost in \eqref{plant_MSE} is the SOI coding scheme in Theorem~\ref{thm1}, and the optimal control signal is
\begin{equation}\label{ZX}
Z_t = -\hat X_t.
\end{equation}
\end{proposition}
\begin{proof}
Given the received codewords $U^i$ and the fact that the next codeword has not been transmitted at $t<\tau_{i+1}$, the optimal control signal $Z_t$ that minimizes \eqref{plant_MSE} is indeed the optimal MMSE decoding policy $\hat X_t$ in \eqref{opt_dec}. Substituting \eqref{plant} and \eqref{ZX} into \eqref{plant_MSE}, we obtain the following MSE,
\begin{align}
   \frac{1}{T}\mathbb E\left[\int_{0}^T (X_t-\hat X_t)^2dt\right],
\end{align}
which is the same as \eqref{MSE}. Thus, the problem of finding the optimal encoding policy in this rate-constrained control system reduces to the problem that we solved in Section~\ref{iv}, whose result is given by Theorem~\ref{thm1}.
\end{proof}
Under the optimal control policy in Proposition~\ref{prop_contr}, the optimal encoder does not rely on the control signal to decide the codeword generating times.

In the traditional stochastic differential equation (SDE) formulation \cite{1}, \cite{Duncan}--\cite{touzi}, the evolution of the plant is described as
\begin{align}\label{SDE_plant}
    \mathrm{d}Y_t = \mathrm{d}X_t + L_t\mathrm{d}t,
\end{align}
where $L_t$ is the control signal.
The state evolutions \eqref{plant} and~\eqref{SDE_plant} are the same if and only if the control signals in~\eqref{plant} and \eqref{SDE_plant} are related as
\begin{align}\label{LZ}
    \int_{0}^t L_s \mathrm{d}s = Z_t,~ \forall t\in[0,T].
\end{align}
Any state evolution described by \eqref{SDE_plant} can be written in the form of \eqref{plant} by setting $Z_t$ as in \eqref{LZ}. Conversely, a state evolution described by \eqref{plant} can be written as \eqref{SDE_plant} if and only if $Z_t$, when viewed as a function of $t$, is almost surely \emph{generalized absolutely continuous in the restricted sense} ($ACG_*$) between any consecutive discontinuous points of $\{Z_t\}_{t=0}^T$ \cite{Talvila}\cite{Gordon}. This is because control signal $L_s$ in~\eqref{LZ} is well-defined if and only if $Z_t^*$ satisfies the $ACG_*$ property.
The function $f\colon [a,b]\rightarrow \mathbb R$ is said to be $ACG_*$ \cite{Talvila}\cite{Gordon} over set $\mathcal E\subset [a,b]$ if $f$ is continuous, and $\mathcal E$ is a countable union of sets $\mathcal E_n$ on each of which $f$ satisfies the following: for each $\epsilon>0$, there exists $\delta>0$ such that $\sum_{i=1}^k \sup_{x,y\in[x_i,y_i]}|F(x)-F(y)|<\epsilon$ for
all finite sets of disjoint open intervals $\{(x_i,y_i)\}_{i=1}^k$ with endpoints in $\mathcal E_n$ and $\sum_{i=1}^k|x_i-y_i|<\delta$.
For example, for stochastic processes of the form $X_t = g_1(t)W_{g_2(t)}+g_3(t)$, the optimal control signal $\{Z_t\}_{t=0}^T$ \eqref{ZX} almost surely satisfies the $ACG_*$ property. Here, $g_1(\cdot), g_3(\cdot)$ are continuous and differentiable except perhaps on a countable set, $g_2(\cdot)$ is continuous, positive, and non-decreasing, and $\{W_t\}_{t=0}^T$ is the Wiener process. In Appendix~\ref{pf_LZ}, we show how to recover $\{L_t\}_{t=0}^T$ from $\{Z_t\}_{t=0}^T$ using \eqref{LZ}, provided that $\{Z_t\}_{t=0}^T$ satisfies the $ACG_*$ property.

\AA str\"{o}m and Bernhardsson \cite{1} considered the controlled system in \eqref{SDE_plant} with $X_t=W_t$ and proposed a control policy that injects an impulse control to drive $Y_t$ to zero once $|Y_t|$ exceeds a threshold. The control signal $L_t$ corresponding to the optimal $Z_t$ in \eqref{ZX} for $X_t= W_t$ recovers \AA str\"{o}m and Bernhardsson's impulse control policy \cite{1} for the Wiener disturbance (Appendix~\ref{pf_LZ}).

\section{Conclusion and discussion}\label{Sec_Conc}
We have studied the optimal rate-constrained causal code for a class of continuous processes satisfying regularity conditions (i)--(iii). Prior art on remote estimation and optimal scheduling mostly considered a sampling frequency constraint, whereas in this work, we introduce a rate constraint. We leverage the information-theoretic framework of our prior work \cite{Nian_2} to establish the jointly optimal causal sampling and quantization policies. We show that the optimal frequency-constrained causal sampling policy is a symmetric threshold sampling policy (Theorems~\ref{prop2}--\ref{prop3}). Prior work \cite{Rabi}--\cite{Nar} on finding the optimal frequency-constrained sampling policy for the Wiener and the OU processes conjectured that the optimal decoding policy is the MMSE decoding policy in \eqref{opt_dec_s}. We confirm that conjecture in Corollary~\ref{cor1}. We show that the optimal causal code is the SOI code that transmits $1$-bit codewords as frequently as possible at the stopping times decided by the optimal frequency-constrained sampling policy (Theorem~\ref{thm1}). Theorems~\ref{prop2}~and~\ref{thm1} demonstrate that the optimal causal code can be easily obtained once we know the optimal sampling policy, revealing the close connection between the frequency-constrained and rate-constrained causal sampling problems. We show that the SOI code minimizes the mean-square cost between the desirable state $0$ and the state of the stochastic plant driven by a process satisfying conditions (i)--(iii). 

Causal rate-constrained sampling for communication over a digital-input noisy channel remains an interesting direction for future research. It is a joint source-channel coding problem that is extremely sensitive to coding delay. Channel codes that can quickly incorporate newly arrived bits into a continuing transmission like the one we developed in \cite{Nian_3} will be instrumental for making progress in this direction.

\appendices
\section{Sufficient condition for (v)}\label{pf_rmk_plus}
Before we show the sufficient condition in Proposition~\ref{prop4'} below, we first characterize the causal sampling policy in Definition~\ref{fcode}. 

Any causal sampling policy in Definition~\ref{fcode} can be characterized by a set-valued process we term \emph{sampling-decision process}. It is a
$\mathcal B_{\mathbb R}$-valued process $\{\mathcal P_t\}_{t=0}^T$ adapted to $\{\mathcal F_t\}_{t=0}^T$,
which decides the stopping times
\begin{equation}\label{gs'}
    \tau_{i+1} = \inf\{t\geq \tau_i\colon \tilde X_t\notin \mathcal P_t\},
\end{equation} 
where the mean-square residual error process $\{\tilde X_t\}_{t=0}^T$ in \eqref{gs'} is defined as
\begin{align}\label{tildeX_appen}
    \tilde X_t \triangleq X_t - \mathbb E[X_t|X_{\tau_i},\tau_i], ~\forall t\in[\tau_i,\tau_{i+1}).
\end{align}
Given any sampling policy $\tau_1,\tau_2,\dots$ and a realization of the process up to time $t$, we can set
\begin{align}\label{Pt_any}
    \mathcal P_t = \begin{cases}
    \mathcal A_t, &  t\neq \tau_i,~ i=1,2,\dots,\\
    \mathcal A_t^c, & t=\tau_i,~ i=1,2,\dots,
    \end{cases} 
\end{align}
where $\mathcal A_t$ is any Borel set the realization of $\tilde X_t$ belongs to. Without assumption (v), $\mathcal P_t$ for $t\in[\tau_i,\tau_{i+1})$ can depend on the input process $\{X_s\}_{s=0}^t$ up to time $t$.
Under assumption (v), $\mathcal P_t$ for $t\in[\tau_i,\tau_{i+1})$ can only depend on the stopping time $\tau_i$ and $\{\tilde X_s\}_{s=\tau_i}^t$ \eqref{tildeX_appen}.

We proceed to present the sufficient condition on the stochastic process under which the optimal sampling policy satisfies (v). We define notations that will be used in Proposition~\ref{prop4'} below. Consider a sampling-decision process $\{\mathcal P_t\}_{t=\tau_k}^T$ with stopping times $\tau_k,\tau_{k+1},\dots$, the mean-square residual error $\tilde X_t$ \eqref{tildeX_appen}, and the MMSE decoding policy $\bar X_t$~\eqref{opt_dec_s}. The value of $\{\mathcal P_t\}_{t=\tau_k}^T$ at time $t\in[\tau_k,T]$ only depends on $\{X_s-\mathbb E[X_s|X_{\tau_k},\tau_k]\}_{s=\tau_k}^t$ and $\tau_k$, i.e.,
\begin{align}\label{Pt34}
    \mathcal P_t = \mathcal P_t(\{X_s-\mathbb E[X_s|X_{\tau_k},\tau_k]\}_{s=\tau_k}^t,\tau_k),~t\in[\tau_k,T].
\end{align}
Denote by $\Pi_{[\tau_k,T]}$ the set of all sampling-decision processes of the form \eqref{Pt34}.
As a result, the stopping times associated with $\{\mathcal P_t\}_{\tau_k}^T \in \Pi_{[\tau_k,T]}$ only satisfy (v) at $i=k$.
Let $N\left(\{\mathcal P_t\}_{t=\tau_k}^T\right)$ represent the number of samples taken between $[\tau_k,T]$ under $\{\mathcal P_t\}_{t=\tau_k}^T$. We denote
\begin{align}\label{drphi_def}
    \underline{D}_r(\phi)\triangleq  \min_{\substack{\{\mathcal P_t\}_{t=r}^T\in \Pi_{[r,T]}\colon\\ \frac{1}{T}\mathbb E[N\left(\{\mathcal P_t\}_{t=r}^T\right)|\tau_k=r]\leq \phi}}\frac{1}{T}\mathbb E\left[\int_{\tau_k}^{T}(X_t-\bar X_t)^2dt\middle |\tau_k=r\right].
\end{align}

Consider an arbitrary sampling-decision process $\{\mathcal P_t'\}_{t=0}^T$~\eqref{gs'} with stopping times $\tau_1',\tau_2',\dots$, the mean-square residual error $\tilde X_t'$, and the MMSE decoding policy $\bar X_t'$. The value of the sampling-decision process $\{\mathcal P_t'\}_{t=0}^T$ at time $t$ can depend on $\{X_s\}_{s=0}^t$, i.e., for all $t\in[\tau_k',T]$,
\begin{align}\label{Ptprime}
    \mathcal P_t' = \mathcal P_t'\left(\{X_s\}_{s=0}^{\tau_k'},\left\{X_s-\mathbb E[X_s|X_{\tau_k'},\tau_k']\right\}_{s=\tau_k'}^t,\tau_k'\right).
\end{align} 
Denote by $\Pi_{[\tau_k',T]}'$ the set of all sampling-decision processes of the form \eqref{Ptprime}.

\begin{proposition}\label{prop4'}
For a stochastic process $\{X_t\}_{t=0}^T$ satisfying (i)--(iii), if $\underline{D}_r(\phi)$ in \eqref{drphi_def} is a convex function in $\phi$ for all $k=0,1,\dots$ and $r\in[0,T]$, then the optimal sampling policy satisfies (v).
\end{proposition}
\begin{proof}
Fix an arbitrary sampling-decision process $\{\mathcal P_t'\}_{t=\tau_k'}^T\in~\Pi_{[\tau_k',T]}'$ at $\tau_k'=r$. To show that the optimal sampling policy of $\{X_t\}_{t=0}^T$ satisfies (v), it suffices to show that for all $k=0,1,\dots$,
 $\underline D_r\left(\frac{1}{T}\mathbb E[N(\{\mathcal P_t'\}_{t=r}^T)|\tau_k'=r]\right)$ is no larger than the MSE achieved by $\{\mathcal P_t'\}_{t=r}^T$, i.e.,
 \begin{align}\nonumber
 &\mathbb E\left[\frac{1}{T}\int_{\tau_k'}^T(X_t-\bar X_t')^2 dt\middle |\tau_k'=r\right] \\ \label{MSE_41}
    \geq~& \underline D_r\left(\frac{1}{T}
    \mathbb E[N(\{\mathcal P_t'\}_{t=r}^T)|\tau_k'=r]\right).
 \end{align}

We fix an arbitrary realization of $\{X_s\}_{s=0}^r=x$ that leads to $\tau_k'=r$, and we construct $\{\mathcal P_t\}_{t=r}^T$ as
\begin{align}\label{Pconstr}
    \mathcal P_t = \mathcal P_t'\left(x,\left\{X_s-\mathbb E[X_s|X_{r},r]\right\}_{s=r}^t,r\right).
\end{align}
The sampling-decision process $\{\mathcal P_t\}_{t=r}^T$ \eqref{Pconstr} satisfies the minimization constraint in \eqref{drphi_def} with 
\begin{align}\label{phi_MSE}
    \phi = \frac{1}{T}\mathbb E[N(\{\mathcal P_t'\}_{t=r}^T)|\{X_s\}_{s=0}^r=x,\tau_k'=r]
\end{align}
due to the reasons that follow.
The process $\{\mathcal P_t\}_{t=r}^T$ \eqref{Pconstr} belongs to $\Pi_{[r,T]}$ since it samples the input process after time $r$ as if it has observed $\{X_s\}_{s=0}^r = x$ regardless of the actual realization of $\{X_s\}_{s=0}^r$. Since $\{\tilde X_t\}_{t\geq \tau_k}$, at $\tau_k=r$, is independent of $\mathcal F_r$ by (iii-a), and $\tau_{i+1}$, $i\geq k$, is conditionally independent of $\{X_s\}_{s=0}^r$ given $\tau_k=r$ due to $\{\mathcal P_t\}_{t=r}^T\in\Pi_{[r,T]}$, we conclude that under $\{\mathcal P_t\}_{t=r}^T$, the random process $\{X_t - \bar X_t\}_{t=r}^T$ conditioned on $\tau_k=r$ has the same probability distribution as $\{X_t - \bar X_t'\}_{t=r}^T$ under $\{\mathcal P_t'\}_{r=0}^T$ conditioned on $\{X_s\}_{s=0}^r=x,\tau_k'=r$. This implies that $\{\mathcal P_t\}_{t=r}^T$ \eqref{Pconstr} achieves average sampling frequency $\phi$ \eqref{phi_MSE}, and that
\begin{subequations}
\begin{align}\nonumber
   &\mathbb E\left[\int_{\tau_k'}^{T}(X_t-\bar X_t')^2dt\middle |\{X_s\}_{s=0}^r=x,\tau_k'=r\right]\\
   =~&\mathbb E\left[\int_{\tau_k}^{T}(X_t-\bar X_t)^2dt\middle |\{X_s\}_{s=0}^r,\tau_k=r\right]\\
   =~&\mathbb E\left[\int_{\tau_k}^{T}(X_t-\bar X_t)^2dt\middle |\tau_k=r\right]\\ \label{MSE47}
   \geq~& \underline D_r\left(\frac{1}{T}\mathbb E[N(\{\mathcal P_t'\}_{t=r}^T)|\{X_s\}_{s=0}^r=x,\tau_k'=r]\right),
\end{align}
\end{subequations}
where \eqref{MSE47} holds because $\{\mathcal P_t\}_{t=\tau_k}^T\in\Pi_{[\tau_k,T]}$. Since \eqref{MSE47} holds for an arbitrary realization of $\{X_s\}_{s=0}^r$ compatible with $\tau_k'=r$, it holds almost surely that
\begin{align}\label{Estar}
    &\mathbb E\left[\int_{\tau_k'}^{T}(X_t-\bar X_t')^2dt\middle |\{X_s\}_{s=0}^r,\tau_k'=r\right]\\ \nonumber
    \geq~& \underline D_r\left(\frac{1}{T}\mathbb E[N(\{\mathcal P_t'\}_{t=r}^T)|\{X_s\}_{s=0}^r,\tau_k'=r]\right).
\end{align}
Taking an expectation of \eqref{Estar}, we conclude 
\begin{align}
    &\mathbb E\left[\frac{1}{T}\int_{\tau_k'}^T(X_t-\bar X_t')^2 dt\middle |\tau_k'=r\right] \\ \geq~& \mathbb E\left[\underline D_r\left(\frac{1}{T}\mathbb E[N(\{\mathcal P_t'\}_{t=r}^T)|\{X_s\}_{s=0}^r,\tau_k'=r]\right)\middle | \tau_k'=r\right],
\end{align}
and \eqref{MSE_41} follows via Jensen's inequality.
\end{proof}

\section{Proof of Theorem~\ref{prop2}}\label{pfprop2}
\subsection{Tools}
We first introduce Lemmas~\ref{l2}--\ref{l5} that supply majorization and real induction tools for proving Theorem~\ref{prop2}. 

Function $f$ \emph{majorizes} $g$, $f\succ g$, if and only if for any Borel measurable set  $\mathcal B\in \mathcal B_{\mathbb R}$ with finite Lebesgue measure, there exits a Borel measurable set $\mathcal A\in \mathcal B_{\mathbb R}$ with the same Lebesgue measure, such that \cite{Lipsa}
\begin{equation}\label{def_major}
  \int_{\mathcal B} g(x)dx \leq \int_{\mathcal A} f(x)dx.  
\end{equation}

Function $f:\mathbb R\rightarrow\mathbb R$ is \emph{even} if $f(x) = f(-x)$ for all $x\in\mathbb R$.

Function $f:\mathbb R\rightarrow\mathbb R$ is \emph{quasi-concave} if for all $x,y\in\mathbb R$, $0\leq \lambda\leq 1$,
\begin{equation}
    f(\lambda x + (1-\lambda)y) \geq \min\{f(x),f(y)\}.
\end{equation}
We denote by $\mathbbm{1}_{(a,b)}(x)$ an indicator function that is equal to $1$ if and only if $x\in(a,b)$.

Lemmas~\ref{l2}--\ref{l4}, stated next, show several majorization properties of pdfs.
\begin{lemma}(\cite[Lemma 2]{Lipsa})\label{l2}
Fix two pdfs $f_X$ and $g_X$, such that $f_X$ is even and quasi-concave and $f_X\succ g_X$. Fix a scalar $c>0$, and a function $h\colon\mathbb R\rightarrow [0,1]$, such that 
\begin{equation}\label{deno_cond}
    \int_{\mathbb R}f_X(x)\mathbbm{1}_{(-c,c)}(x)dx=\int_{\mathbb R}g_X(x)h(x)dx,
\end{equation}
Then,
\begin{equation}
    f_{X|X\in(-c,c)}\succ g_X',
\end{equation}
where the pdfs $f_{X|X\in(-a,a)}$ and $g_X'$ are given by,
\begin{equation}
\begin{aligned}
    &f_{X|X\in(-c,c)}(x)=\frac{f_X(x)\mathbbm{1}_{(-c,c)}(x)}{\int_{\mathbb R}f_X(x)\mathbbm{1}_{(-c,c)}(x)dx}\\
    &g_X'(x) = \frac{g_X(x)h(x)}{\int_{\mathbb R}g_X(x)h(x)dx}.
\end{aligned}
\end{equation}

\end{lemma}
\begin{lemma}(\cite[Lemma 6.7]{Hj})\label{l3}
Fix two pdfs $f_X$ and $g_X$, such that $f_X$ is even and quasi-concave and that $f_X$ majorizes $g_X$, $f_X\succ g_X$. Fix an even and quasi-concave pdf $r_Y$. Then, the convolution of $f_X$ and $r_Y$ majorizes the convolution of $g_X$ and $r_Y$,
\begin{equation}
    f_X*r_Y\succ g_X*r_Y,
\end{equation}
Furthermore, $f_X*r_Y$ is even and quasi-concave.

\end{lemma}
\begin{lemma}(\cite[Lemma 4]{Lipsa})\label{l4}
Fix two pdfs $f_X$ and $g_X$ such that $f_X$ is even and quasi-concave and that $f_X$ majorizes $g_X$, $f_X\succ g_X$. Then,
\begin{equation}
    \int_{\mathbb R}x^2f_X(x)dx\leq \int_{\mathbb R}(x-y)^2g_X(x)dx,~\forall y\in\mathbb R.
\end{equation}
\end{lemma}

Lemma~\ref{l5}, stated next, provides a mathematical proof technique called \emph{real induction}. We will use it to prove that the assertions in Lemma~\ref{l6'}, stated below, hold on a continuous interval.
\begin{lemma}\label{l5}(Real induction \cite[Thm. 2]{Pete}) A subset $S\subset [a,b]$, $a<b$ is called inductive if
\begin{itemize}
    \item[1)] $a\in S$;
    \item[2)] If $a\leq x<b$, $x\in S$, then there exists $y>x$ such that $[x,y]\in S$;
    \item[3)] If $a\leq x<b$, $[a,x)\in S$, then $x\in S$.
\end{itemize} 
If a subset $S\subset [a,b]$ is inductive, then $S=[a,b]$.
\end{lemma}

\subsection{A technical lemma}\label{tech_lemma}
We define the following notations for two sampling-decision processes $\{\mathcal P_t\}_{t=0}^T$ and $\{\mathcal P^{\mathrm{sym}}_t\}_{t=0}^T$ (see Appendix~\ref{pf_rmk_plus}). Fix an arbitrary sampling-decision process $\{\mathcal P_t\}_{t=0}^T$ \eqref{gs'} satisfying (iv)--(v). It gives rise to a sampling policy with stopping times $\tau_1,\tau_2,\dots$ via \eqref{gs'}. We recall the definition of the mean-square residual error (MSRE) process $\{\tilde X_t\}_{t=0}^T$ in (iii) and denote the MSRE process under $\{\mathcal P_t\}_{t=0}^T$ as
\begin{subequations}\label{Yta}
\begin{align}
    \tilde X_t~& = \tilde X_t(\{\mathcal P_s\}_{s=0}^T)\\
    &\triangleq X_t-\mathbb E[X_t|X_{\tau_i},\tau_i], t\in[\tau_i,\tau_{i+1}).
\end{align}
\end{subequations}
We define the residual error estimate (REE) process $\{\bar {\tilde X}_t\}_{t=0}^T$ under $\{\mathcal P_t\}_{t=0}^T$ as
\begin{subequations}\label{Ytbar'}
\begin{align}
   \bar {\tilde{X}}_t~&=\bar {\tilde{X}}_t(\{\mathcal P_s\}_{s=0}^T)\\\label{Ytbar} 
   &\triangleq \bar X_t - \mathbb E[X_t|X_{\tau_i},\tau_i]\\ \label{XG'b}
   &=\mathbb E[\tilde X_t|\{X_{\tau_j}\}_{j=1}^i,{\tau}^{i}, t<\tau_{i+1}]\\\label{XG'c}
   &=\mathbb E[\tilde X_t| \tau_i,t<\tau_{i+1}],~t\in[\tau_i,\tau_{i+1}),
\end{align}
\end{subequations}
where $\bar X_t=\bar X_t(\{\mathcal P_s\}_{s=0}^T)$ is the MMSE decoding policy defined in \eqref{opt_dec_s}; the equality in \eqref{XG'b} holds since $\mathbb E[X_t|X_{\tau_i},\tau_i]\in\sigma(\{X_{\tau_j}\}_{j=1}^i,\tau^i, t<\tau_{i+1})$; \eqref{XG'c} holds because $\tilde X_t$ is independent of $\{X_{\tau_j}\}_{j=1}^i,\tau^{i}$ due to (iii-a), and the event $\{t<\tau_{i+1}\}$ is independent of $\{X_{\tau_j}\}_{j=1}^i,\tau^{i-1}$ given $\tau_i$ due to (v). We recall that $N(\{\mathcal P_t\}_{t=0}^T)$ defined above Proposition~\ref{prop4'} in Appendix~\ref{pf_rmk_plus} represents the number of stopping times in $[0,T]$, and we simplify this notation as
\begin{align}
    N \triangleq N(\{\mathcal P_t\}_{t=0}^T).
\end{align}
We denote the left-closed continuous interval
\begin{align}\label{omega}
    \Omega_{\tau_{i+1}}(s) \triangleq \{t\in[s,T]\colon &\mathbb P[\tau_{i+1}>t|\tau_i=s] >0\}, 
\end{align}
for all $s\in\mathrm{Supp}(f_{\tau_i})$, and the left-open continuous interval
\begin{align}
    \bar \Omega_{\tau_{i+1}}(s) \triangleq \Omega_{\tau_{i+1}}(s)\setminus\{s\}.
\end{align}

Given $\{\mathcal P_t\}_{t=0}^T$, we construct a sampling-decision process $\{\mathcal P^{\mathrm{sym}}_t\}_{t=0}^T$ \eqref{gs'} of the form \eqref{dstp}, which via \eqref{gs'} is associated with a sampling policy with stopping times $\tau_1',\tau_2',\dots$ , such that the symmetric thresholds $\{a_i(r,s)\}_{r=s}^{T}$ of $\{\mathcal P^{\mathrm{sym}}_t\}_{t=0}^T$ satisfy for all $s\in\mathrm{Supp}(f_{\tau_i})$, $t\in[s,T]$,
\begin{equation}
\begin{aligned}\label{constr2}
&\mathbb P[\tilde X_r'\in(-a_i(r,s),a_i(r,s)),\forall r\in[s,t]|\tau_i'=s]\\ =~&\mathbb P[\tau_{i+1}>t|\tau_i=s]. 
\end{aligned}
\end{equation}
This is possible since by adjusting the thresholds, the left side of \eqref{constr2} can be equal to any non-increasing function in $t$ bounded between $[0,1]$. Under $\{\mathcal P^{\mathrm{sym}}_t\}_{t=0}^T$ \eqref{constr2}, for all $s\in\mathrm{Supp}(f_{\tau_i})$, $i=1,2,\dots$, it holds that
\begin{align}
    \Omega_{\tau_i}(s) = \Omega_{\tau_i'}(s),\\
    \bar \Omega_{\tau_i}(s) = \bar \Omega_{\tau_i'}(s).
\end{align}
We denote the MSRE and the REE processes and the number of stopping times on $[0,T]$ under $\{\mathcal P_t^{\mathrm{sym}}\}_{t=0}^T$ respectively by
\begin{align}\label{Yt'barY'N'}
         &\tilde X_t' = \tilde X_t(\{\mathcal P^{\mathrm{sym}}_s\}_{s=0}^T),\\\label{YAc}
         &\bar {\tilde{X}}_t' = \tilde {\bar X}_t(\{\mathcal P^{\mathrm{sym}}_s\}_{s=0}^T)= 0,\\  
    &N' = N(\{\mathcal P^{\mathrm{sym}}_s\}_{s=0}^T),
\end{align}
where \eqref{YAc} holds since we can write $\bar {\tilde{X}}_t'$ as \eqref{XG'c} with $\tau_i$ replaced by $\tau_i'$ using the argument that justifies \eqref{XG'c}; $\tilde X_t'$ has an even and quasi-concave pdf due to the assumption (iii-b), and the pdf of $\tilde X_t$ conditioned on $\tau_i',t<\tau_{i+1}'$ under a symmetric threshold sampling-decision process of the form~ \eqref{dstp} is still even and quasi-concave. 

We denote the following probabilities
\begin{subequations}
\begin{align}\label{QP_denote1}
&\mathbb{Q}_i(a,b,c,d) \triangleq \mathbb P[\tau_{i+1}>a|\tau_{i+1}>b, \tau_i = c, \tilde X_a = d]\\ \label{QP_denote2}
&\mathbb{Q}_i'(a,b,c,d) \triangleq \mathbb P[\tau_{i+1}'>a|\tau_{i+1}'>b, \tau_i' = c, \tilde X_a' = d].
\end{align}
\end{subequations}

We proceed to introduce Lemma~\ref{l6'} using the notations defined in \eqref{Yta}--\eqref{QP_denote2}. We will use the assertions in Lemma~\ref{l6'} to compare the MSEs achieved by $\{\mathcal P_t\}_{t=0}^T$ and $\{\mathcal P_t^{\mathrm{sym}}\}_{t=0}^T$. 

\begin{lemma}\label{l6'}
The pdfs $f_{\tilde X_t|\tau_i=s,\tau_{i+1}>t}$ and $f_{\tilde X_t'|\tau_i'=s,\tau_{i+1}'>t}$ exist for all $s\in\mathrm{Supp}(f_{\tau_i})$,  $t\in\bar\Omega_{\tau_{i+1}}(s)$. Furthermore, for all $s\in\mathrm{Supp}(f_{\tau_i})$, $t\in\bar\Omega_{\tau_{i+1}}(s)$, it holds that
\begin{align}\label{maji1}
   f_{\tilde X_t'|\tau_i'=s,\tau_{i+1}'>t}\succ f_{\tilde X_t|\tau_i=s,\tau_{i+1}>t},\\ \label{maji1_even}
    f_{\tilde X_t'|\tau_i'=s,\tau_{i+1}'>t}~\text{is even and quasi-concave}.
\end{align}
\end{lemma}

\begin{proof}[Proof of Lemma~\ref{l6'}]
We prove that $f_{\tilde X_t|\tau_i=s,\tau_{i+1}>t}$ exists. The proof that $f_{\tilde X_t'|\tau_i'=s,\tau_{i+1}'>t}$ exists is similar. Since $\tilde X_t$ at $t\geq \tau_i=s$, is independent of $\mathcal F_s$ by (iii-a) and is equal to $R_t(s,s)$ by (iii-b), we compute $f_{\tilde X_t|\tau_i=s,\tau_{i+1}>s}$ using \eqref{Yt_r},
\begin{equation}\label{thm1_70}
    f_{\tilde X_t|\tau_i=s,\tau_{i+1}>s} = f_{R_t(s,s)}.
\end{equation}
Thus, $f_{\tilde X_t|\tau_i=s,\tau_{i+1}>s}$ exists since $f_{R_t(s,s)}$ is a valid pdf by (iii-b). To establish that $f_{\tilde X_t|\tau_i=s,\tau_{i+1}>t}(y)$ exists, we compute
\begin{subequations}\label{54_behind}
\begin{align}\label{54_b1}
 &f_{\tilde X_t|\tau_i=s,\tau_{i+1}>t}(y)
  = f_{\tilde X_t|\tau_i=s,\tau_{i+1}>s,\tau_{i+1}>t}(y)\\\label{54_b2}
  =~& \frac{\mathbb Q_i(t,s,s,y)f_{\tilde X_t|\tau_i=s,\tau_{i+1}>s}(y)}{\mathbb P[\tau_{i+1}>t|\tau_i=s,\tau_{i+1}>s]},
  \end{align}
\end{subequations}
where \eqref{54_b1} holds since $\tau_{i+1}>t$ implies $\tau_{i+1}>s$. In \eqref{54_b2}, we observe that for all $t\in\bar \Omega_{\tau_{i+1}}(s)$, the pdf $f_{\tilde X_t|\tau_{i+1}>s,\tau_i=s}$ exists by \eqref{thm1_70}; the denominator of \eqref{54_b2} is nonzero. We conclude that the pdf $f_{\tilde X_t|\tau_{i}=s,\tau_i>t}$ exists for all $s\in\mathrm{Supp}(f_{\tau_i})$, $t\in\bar\Omega_{\tau_{i+1}}(s)$.

The assertion \eqref{maji1} holds if and only if 
\begin{itemize}
\item [(a)] for all $s\in\mathrm{Supp}(f_{\tau_i})$, $t\in\bar\Omega_{\tau_{i+1}}(s)$ and for any Borel measurable set $\mathcal B\in\mathcal B_{\mathbb R}$ with finite Lebesgue measure, there exists a Borel measurable set $\mathcal A\in\mathcal B_{\mathbb R}$ with the same Lebesgue measure, such that
 \begin{equation}\label{maji1f}
 \begin{aligned}
  &\mathbb P[\tilde X_t'\in\mathcal A|\tau_i'=s,\tau_{i+1}'>t]\\
  \geq~&  \mathbb P[\tilde X_t\in\mathcal B|\tau_i=s,\tau_{i+1}>t],
  \end{aligned}
\end{equation}
\end{itemize}
holds. This is because \eqref{maji1f} is a rewrite of \eqref{maji1} using the definition of majorization \eqref{def_major}.

The assertion \eqref{maji1_even} holds if and only if for all $s\in\mathrm{Supp}(f_{\tau_i})$, $t\in\bar\Omega_{\tau_{i+1}}(s)$, all of the following hold:
\begin{itemize}
\item[(b)] the conditional cdf $\mathbb P[\tilde X_t'\leq y|\tau_i'=s,\tau_{i+1}'>t]$ is convex for $y<0$ and is concave for $y>0$. 
\item[(c)] for any $y> 0$, 
\begin{equation}\label{maji1_evenf}
\begin{aligned}
    &\mathbb P[\tilde X_t'\in(0,y]|\tau_i'=s,\tau_{i+1}'>t]\\
    =&\mathbb P[\tilde X_t'\in[-y,0)|\tau_i'=s,\tau_{i+1}'>t].
\end{aligned}
\end{equation}
\end{itemize}
This is because $f_{\tilde X_t'|\tau_i'=s,\tau_{i+1}'>t}$ is quasi-concave if and only if (b) holds, and $f_{\tilde X_t'|\tau_i'=s,\tau_{i+1}'>t}$ is even if and only if (c) holds. 

Items (a)--(c) facilitate proving that the assertions \eqref{maji1}--\eqref{maji1_even} hold on the left-open interval $\bar\Omega_{\tau_{i+1}}(s)$. Real induction, which must be used on a left-closed interval, does not apply to show~\eqref{maji1}--\eqref{maji1_even} directly, since the densities in \eqref{maji1}--\eqref{maji1_even} do not exist at $t=s$. Instead, we apply real induction to show (a)--(c). Using real induction in Lemma~\ref{l5}, we verify that conditions 1), 3), 2) in Lemma~\ref{l5} hold for (a)--(c) in on $t\in\Omega_{\tau_{i+1}}(s)$ one by one.

To verify that the condition 1) in Lemma~\ref{l5} holds, we need to show that (a)--(c) hold for $t=s$. This is trivial since
\begin{equation}\label{cdf_com}
\begin{aligned}
&\mathbb P[\tilde X_s'=0|\tau_i'=s,\tau_{i+1}'>s]\\
    =~&\mathbb P[\tilde X_s=0|\tau_i=s,\tau_{i+1}>s]\\
    =~&1.   
\end{aligned}
\end{equation}

Next, we show that condition 3) in Lemma~\ref{l5} holds, that is, assuming that (a)--(c) hold for all $t\in[s,r)$, $r\in\bar \Omega_{\tau_{i+1}}(s)$, we prove that (a)--(c) hold for $t=r$. Equivalently, we show that~\eqref{maji1}--\eqref{maji1_even} hold for $t=r$. Let $\delta \in(0,r-s]$. At time $t=r$, we calculate the left side of \eqref{maji1} as \begin{subequations}\label{FA0}
\begin{align}\nonumber
    &f_{\tilde X_r'|\tau_i'=s,\tau_{i+1}'>r}(y)\\ \label{p2_32b}
=~&\lim_{\delta\rightarrow0^+}f_{\tilde X_r'|\tau_i'=s,\tau_{i+1}'> r-\delta,\tau_{i+1}'>r}(y)\\\label{p2_32c}
    =~&\lim_{\delta\rightarrow0^+}\frac{\mathbb Q_i'(r, r-\delta,s,y)f_{\tilde X_r'|\tau_i'=s,\tau_{i+1}'>r-\delta}(y)}{\int_{\mathbb R}\mathbb Q_i'(r, r-\delta,s,y)f_{\tilde X_r'|\tau_i'=s,\tau_{i+1}'>r-\delta}(y)dy}\\ \label{p2_32d}
    =~&\lim_{\delta\rightarrow0^+}\frac{\mathbbm{1}_{(-a_i(r,s),a_i(r,s))}(y)f_{\tilde X_r'|\tau_i'=s,\tau_{i+1}'>r-\delta}(y)}{\int_{\mathbb R}\mathbbm{1}_{(-a_i(r,s),a_i(r,s))}(y)f_{\tilde X_r'|\tau_i'=s,\tau_{i+1}'>r-\delta}(y)dy},
\end{align}
\end{subequations}
where \eqref{p2_32b} holds since the event $\tau_{i+1}'>r$ implies the event $\tau_{i+1}'>r-\delta$; the pdf $f_{\tilde X_r'|\tau_i'=s,\tau_{i+1}'>r-\delta}$ in \eqref{p2_32c} exists since \eqref{54_behind} holds with $\tilde X_t$, $\tau_i=s$, $\tau_{i+1}>s$, $\tau_{i+1}>t$ replaced by $\tilde X_r'$, $\tau_i'=s$, $\tau_{i+1}'>s$, $\tau_{i+1}'>r-\delta$, respectively; \eqref{p2_32d} holds since 
\begin{equation}
    \lim_{\delta\rightarrow 0^+}\mathbb Q_i'(r,r-\delta,s,y) = \mathbbm{1}_{(-a_i(r,s),a_i(r,s))}(y).
\end{equation}
Similarly, replacing $\mathbb Q_i'$ in \eqref{p2_32c} by $\mathbb Q_i$, we calculate the right side of \eqref{maji1} as
\begin{align}\nonumber
  &f_{\tilde X_r|\tau_i=s,\tau_{i+1}>r}(y)\\ \label{FG0}
  =~&\lim_{\delta\rightarrow0^+}\frac{\mathbb Q_i(r, r-\delta, s, y)f_{\tilde X_r|\tau_i=s,\tau_{i+1}>r-\delta}(y)}{\int_{\mathbb R}\mathbb Q_i(r, r-\delta, s, y)f_{\tilde X_r|\tau_i=s,\tau_{i+1}>r-\delta}(y)dy},
\end{align}
where the pdf $f_{\tilde X_r|\tau_i=s,\tau_{i+1}>r-\delta}(y)$ exists since \eqref{54_behind} holds with $\tilde X_t$, $\tau_{i+1}>t$ replaced by $\tilde X_r$, $\tau_{i+1}>r-\delta$ respectively.

To check that \eqref{maji1} holds at $t=r$, we first prove that $f_{\tilde X_r'|\tau_i'=s,\tau_{i+1}'>r-\delta}$ majorizes $ f_{\tilde X_r|\tau_i=s,\tau_{i+1}>r-\delta}$. Note that $R_r(r-\delta,s)$ is independent of $\{\tilde X_t\}_{t=0}^{r-\delta}$ due to (iii-a), and thus is independent of the event $\{\tau_{i+1}'>r-\delta,\tau_i'=s\}$. We obtain $\tilde X_r'$ using \eqref{Yt_r},
\begin{equation}\label{maj2}
     f_{\tilde X_r'|\tau_i'=s,\tau_{i+1}'>r-\delta}= f_{q_r(r-\delta)\tilde X_{r-\delta}'|\tau_i'=s,\tau_{i+1}'>r-\delta}*f_{R_r(r-\delta,s)}.
\end{equation}
By \eqref{maj2} and the inductive hypothesis that (a)--(c) holds for $t\in[s,r)$, the assumptions in Lemma~\ref{l3} are satisfied with $f_X\leftarrow f_{q_r(r-\delta)\tilde X_{r-\delta}'|\tau_i'=s,\tau_{i+1}'>r-\delta}$, $g_X\leftarrow f_{q_r(r-\delta)\tilde X_{r-\delta}|\tau_i=s,\tau_{i+1}>r-\delta}$, $r_Y\leftarrow f_{R_r(r-\delta,s)}$. We conclude that
\begin{equation}\label{maj3}
    f_{\tilde X_r'|\tau_i'=s,\tau_{i+1}'>r-\delta}\succ  f_{\tilde X_r|\tau_i=s,\tau_{i+1}>r-\delta},
\end{equation}
\begin{equation}\label{maj3_behine}
  f_{\tilde X_r'|\tau_i'=s,\tau_{i+1}'>r-\delta}~\text{is even and quasi-concave}.
\end{equation}
Due to \eqref{maj3_behine} and the fact that the indicator function in \eqref{p2_32d} is over an interval symmetric about zero, we conclude \eqref{maji1_even} holds for $t=r$. 
By \eqref{constr2}, \eqref{maj3} and \eqref{maj3_behine}, the assumptions in Lemma~\ref{l2} are satisfied with $f_X\leftarrow f_{\tilde X_r'|\tau_i'=s,\tau_{i+1}'>r-\delta}$, $g_X\leftarrow f_{\tilde X_r|\tau_i=s,\tau_{i+1}>r-\delta}$, $f_{X|X\in(-c,c)} \leftarrow f_{\tilde X_r'|\tau_i'=s,\tau_{i+1}'>r}$, and $g_X' \leftarrow f_{\tilde X_r|\tau_i=s,\tau_{i+1}>r}$, $c\leftarrow a_i(r,s)$, $h\leftarrow \mathbb Q_i(r,r-\delta,s,y)$. Thus, we conclude that \eqref{maji1} holds for $t=r$. Therefore, \eqref{maji1}--\eqref{maji1_even} hold for $t=r$, i.e., (a)--(c) hold for $t=r$.

To prove that the condition 2) in Lemma~\ref{l5} holds, we assume (a)--(c) hold for $t=r$, and prove that the following holds:
\begin{subequations}\label{maj_reld}
\begin{align}\label{maj_reld_a}
    &\lim_{\delta\rightarrow 0^+}f_{\tilde X_{r+\delta}'|\tau_i'=s,\tau_{i+1}'>r+\delta}\succ \lim_{\delta\rightarrow0^+}f_{\tilde X_{r+\delta}|\tau_i=s,\tau_{i+1}>{r+\delta}},\\
    &\lim_{\delta\rightarrow 0^+}f_{\tilde X_{r+\delta}'|\tau_i'=s,\tau_{i+1}'>r+\delta}~\text{is even and quasi-concave}.
\end{align}
\end{subequations} 
The right and the left sides of \eqref{maj_reld_a} are equal to \eqref{p2_32d} and~\eqref{FG0} respectively with $r$ replaced by $r+\delta$. It is easy to see that \eqref{maj2}--\eqref{maj3_behine} and the assumptions in Lemma~\ref{l2} hold with $r$ replaced by $r+\delta$. Thus, we conclude that \eqref{maj_reld} holds. 

Using the real induction in Lemma~\ref{l5}, we have shown that (a)--(c) hold for all $s\in\mathrm{Supp}(f_{\tau_i})$, $t\in\Omega_{\tau_{i+1}}(s)$. Thus, \eqref{maji1}--\eqref{maji1_even} hold for all $s\in\mathrm{Supp}(f_{\tau_i})$, $t\in\bar\Omega_{\tau_{i+1}}(s)$. 
\end{proof}

\subsection{Proof of Theorem~\ref{prop2}}\label{prop2_sub_pf}
The sampling-decision process $\{\mathcal P_t^{\mathrm{sym}}\}_{t=0}^T$ leads to the same average sampling frequency as $\{\mathcal P_t\}_{t=0}^T$. This is because \eqref{constr2} implies that for all $s\in\mathrm{Supp}(f_{\tau_i})$, $t\in[s,T]$,
\begin{equation}\label{Pcond}
    \mathbb P[\tau_{i+1}>t|\tau_i=s]=\mathbb P[\tau_{i+1}'>t|\tau_i'=s].
\end{equation}
Together with the Markov property of the stopping times (assumption (v)), \eqref{Pcond} implies that the joint distribution of $\tau_1,\tau_2,\dots$ is equal to the joint distribution of $\tau_1',\tau_2',\dots$ We conclude that $\{\mathcal P_t\}_{t=0}^T$ and $\{\mathcal P_t^{\mathrm{sym}}\}_{t=0}^T$ lead to the same average sampling frequency
\begin{align}
    \mathbb E[N] = \mathbb E[N'].
\end{align}

Next, we show $\{\mathcal P_t^{\mathrm{sym}}\}_{t=0}^T$ achieves an MSE no larger than that achieved by $\{\mathcal P_t\}_{t=0}^T$. Due to \eqref{XG'c}, \eqref{YAc}, and \eqref{maji1}--\eqref{maji1_even} in Lemma~\ref{l6'}, we can apply Lemma~\ref{l4} with $f_X\leftarrow f_{\tilde X_t'|\tau_i'=s,\tau_{i+1}'>t}$ and $g_X\leftarrow f_{\tilde X_t|\tau_i=s,\tau_{i+1}>t}$, yielding
\begin{equation}\label{Econd}
\mathbb E\left[(\tilde X_t-\bar {\tilde{X}}_t)^2|\tau_{i}=s, \tau_{i+1}>t\right]\geq \mathbb E\left[\tilde X_t'^2|\tau_{i}'=s, \tau_{i+1}'>t\right].
\end{equation}
Combining \eqref{Pcond} and \eqref{Econd}, we conclude by law of total expectation that $\{\mathcal P_t^{\mathrm{sym}}\}_{t=0}^T$ achieves an MSE no larger than that achieved by $\{\mathcal P_t\}_{t=0}^T$.

\section{Proof of Corollary~\ref{cor1}}\label{pfcor1}
Under a symmetric threshold sampling policy \eqref{dstp}, the MMSE decoding policy in \eqref{opt_dec_s} can be expanded as, for $\tau_i\leq t<\tau_{i+1}$,
\begin{subequations}\label{ms}
\begin{align}\label{mmse_s_a}
    \bar X_t
    =&\mathbb E[X_t|\{X_{\tau_j}\}_{j=1}^i,\tau^i, t<\tau_{i+1}]\\ \label{mmse_s_b}
 =&\bar {\tilde{X}}_t+ \mathbb E[X_t|X_{\tau_i},\tau_i]\\ \label{mmse_s_d}
    =&\mathbb E[X_t|X_{\tau_i},\tau_i],
\end{align}
\end{subequations}
where $\bar {\tilde{X}}_t$ in \eqref{mmse_s_b} is equal to $\bar {\tilde{X}}_t'$ in \eqref{YAc}, thus is equal to zero.
\section{Proof of Corollary~\ref{cor2}}\label{pfcor2}
Given any causal sampling policy such that \eqref{comm_cons_f} is satisfied with a strict inequality, we construct a causal sampling policy that satisfies \eqref{comm_cons_f} with equality and leads to an MSE no worse than that achieved by the given causal sampling policy. 

Given an arbitrary symmetric threshold sampling policy \eqref{dstp} with stopping times $\tau_1,\tau_2,\dots$, we denote by $N_t$ the number of samples taken in $[0,t]$. Let $t'$, $t'\in(0,T)$ be a dummy deterministic time. We decompose the MSE under the given sampling policy as
\begin{subequations}\label{cor2eq1}
\begin{align}\label{cor2eq1a}
    &\mathbb E\left[\sum_{i=0}^{N_{t'}-1}\int_{\tau_i}^{\tau_{i+1}}(X_t-\mathbb E[X_t|X_{\tau_i},\tau_i])^2 dt\right]\\ \label{cor2eq1c}
    +&\mathbb E\left[\int_{\tau_{N_{t'}}}^{t'}(X_t-\mathbb E[X_t|X_{\tau_{N_{t'}}},\tau_{N_{t'}}])^2 dt\right]\\ \label{cor2eq1d}
    +&\mathbb E\left[\int_{t'}^{\tau_{N_{t'}+1}}(X_t-\mathbb E[X_t|X_{\tau_{N_{t'}}},\tau_{N_{t'}}])^2 dt\right]\\ \label{cor2eq1e}
    +&\mathbb E\left[\sum_{i=N_{t'}+1}^{N_T}\int_{\tau_i}^{\tau_{i+1}}(X_t-\mathbb E[X_t|X_{\tau_i},\tau_i])^2 dt\right],
\end{align}
\end{subequations}
where $\tau_{N_T+1}\triangleq T$.

Under the given sampling policy $\tau_1,\tau_2,\dots$, we construct a sampling policy by inserting an extra deterministic sampling time $t'$. The resultant MSE is the same as \eqref{cor2eq1} with \eqref{cor2eq1d} replaced by
\begin{equation}\label{cor2eq2d}
    \mathbb E\left[\int_{t'}^{\tau_{N_{t'}+1}}(X_t-\mathbb E[X_t|X_{t'}])^2 dt\right],
\end{equation}
since a sample is taken at time $t'$ under the constructed sampling policy. Since
\begin{subequations}
\begin{align}
    &\sigma(X_{\tau_{N_{t'}}},\tau_{N_{t'}})\subseteq \sigma(\mathcal F_{t'})\\ \label{sigmab}
&\mathbb E[X_t|\mathcal F_{t'}] = \mathbb E[X_t|X_{t'}],
\end{align}    
\end{subequations}
where \eqref{sigmab} is due to the strong Markov process (i) in Section~\ref{IA}, we conclude that \eqref{cor2eq1d} $\geq$ \eqref{cor2eq2d}.

Thus, by introducing extra sampling times, we can achieve the same or a lower MSE. We can express the difference between the frequency constraint $F$ and the average sampling frequency under the given sampling policy as
\begin{equation}
    FT-\mathbb E[N_T] = I+D,
\end{equation}
where $I\in \mathbb N$ represents the non-negative integer part, and $D\in(0,1)$ represents the decimal part. By introducing $I$ different deterministic sampling times, we can compensate the integer part $I$. By introducing a random sampling time stamp $t$ with probability $D$ to sample and probability $1-D$ not to sample, we can compensate the decimal part. Therefore, for any sampling policy whose average sampling frequency is strictly less than $F$, we can always construct a sampling policy that achieves the maximum sampling frequency $F$ and leads to an MSE no worse than that achieved by the arbitrarily fixed sampling policy. 

\section{Proof of Corollary~\ref{cor3}}\label{pfcor3}
We show that symmetric thresholds $\{a_i(r,s)\}_{r=s}^T$ in \eqref{constr2} must satisfy \eqref{greater} for all $s\in\mathrm{Supp}(f_{\tau_i})$.

Due to (vi), the probability on the right side of \eqref{constr2} is continuous in $t\in[s,T]$ for all $s\in\mathrm{Supp}(f_{\tau_i})$. Thus, for all $s\in\mathrm{Supp}(f_{\tau_i})$, $t\in[s,T)$,
\begin{subequations}\label{prob_ineq}
\begin{align}
    &\lim_{\delta\rightarrow 0^+} \mathbb P\left[\tilde X_r'\in (-a_i(r,s),a_i(r,s)),\forall r\in[s,t+\delta]\middle | \tau_i'=s\right]\\
    =~&\mathbb P\left[\tilde X_r'\in (-a_i(r,s),a_i(r,s)),\forall r\in[s,t]\middle| \tau_i'=s\right].
\end{align}
\end{subequations}
By the continuity of $\tilde X_r'$ in (iii-b), \eqref{prob_ineq} implies \eqref{greater}.

\section{Proof of Theorem~\ref{prop3}}\label{pfprop3}
First, we introduce Lemma~\ref{l6}, stated next, that will be helpful in proving \eqref{solvet}. Second, we prove that symmetric threshold sampling policies \eqref{dstp} in Theorem~\ref{prop2} can be reduced to \eqref{stp} in the setting of Theorem~\ref{prop3}, i.e., under the assumption that $\{X_t\}_{t\geq 0}$ has time-homogeneous property in Definition~\ref{def_th} and $T=\infty$. Then, we show that Remark~\ref{rmk1} holds and prove that \eqref{solvet} holds using Lemma~\ref{l6}.

\begin{lemma}(e.g. \cite[Proposition 1(ii)]{Mark}) \label{l6}
 Suppose that $Z_0,Z_1,\dots$ are i.i.d. Let $Q_t\triangleq\sum_{i=0}^{\infty}\mathbbm{1}_{[0,t]}\left(\sum_{k=0}^iZ_k\right)$. Let $R_0,R_1,\dots$ be i.i.d rewards, and let $S_t\triangleq\sum_{i=0}^{Q_t}R_i$ be the renewal reward process.  
If $0<\mathbb E[Z_i]<\infty$, $\mathbb E[|R_i|]<\infty$, then
\begin{equation}\label{pf2_75}
    \lim_{T\rightarrow\infty}\frac{\mathbb E[S_T]}{T}=\frac{\mathbb E[R_0]}{\mathbb E[Z_0]}.
\end{equation}
\end{lemma}

Since the stochastic process considered in Theorem~\ref{prop3} is infinitely long, we use the DFF in the infinite time horizon:
\begin{equation}\label{dsf_eq_inf}
    \underline D^{\infty}(F) =\inf_{\substack{\{\mathcal P_t\}_{t\geq0}\in\Pi\colon\\ \eqref{comm_cons_f_b}}} \limsup_{T\rightarrow\infty}\frac{1}{T}\mathbb E\left[\int_{0}^T(\tilde X_t-\bar{\tilde X}_t)^2\right],
\end{equation}
where $\Pi$ is the set of all sampling-decision processes \eqref{gs'} of the form \eqref{dstp} satisfying (iv) and (v) in Section~\ref{IA} over the infinite time horizon. 
Note that for any stopping time $\tau$, and for any $t\geq \tau$, we have
\begin{equation}\label{pf2_c1}
    \{\tilde X_t\}_{t\geq \tau}~\text{and}~\{\tilde X_{t-\tau}\}_{t-\tau\geq 0}~\text{have the same distribution},
\end{equation}
\begin{equation}\label{pf2_c2}
    \{\tilde X_t\}_{t\geq \tau}~\text{is independent of}~ \{\tilde X_t\}_{t=0}^{ \tau},
\end{equation}
where \eqref{pf2_c1} is due to the time-homogeneity of $\{X_t\}_{t\geq 0}$ in Definition~\ref{def_th}, and \eqref{pf2_c2} is due to (iii-a) in Section~\ref{IA}.

Using \eqref{pf2_c1}--\eqref{pf2_c2} and assumption (iv), we will prove that the sampling-decision process that achieves the $\underline{D}^{\infty}(F)$ for time-homogeneous continuous Markov processes satisfying assumptions (i)--(iii) is of the form \eqref{stp}.

Given an arbitrary sampling-decision process $\{\mathcal P_t\}_{t\geq 0}$ of the form \eqref{dstp}, we define its MSRE \eqref{Yta} and REE \eqref{Ytbar'} processes  as
\begin{equation}
\begin{aligned}
    &\tilde X_t \triangleq \tilde X_t(\{\mathcal P_s\}_{s\geq 0}),\\
    &\bar {\tilde{X}}_t \triangleq \bar {\tilde{X}}_t(\{\mathcal P_s\}_{s\geq 0}).
\end{aligned}
\end{equation}
Denote by $\tau_1,\tau_2,\dots$ the stopping times of the causal sampling policy characterized by $\{\mathcal P_t\}_{t\geq 0}$. Assume that the sampling-decision process that achieves $\underline{D}^{\infty}(F)$ \eqref{dsf_eq_inf} is $\left\{\mathcal P_t^{\mathrm{(a)}}\right\}_{t\geq 0}$. We have,
\begin{subequations}\label{pfthm2_71}
\begin{align}\label{pf2_72a}
 &\underline{D}^{\infty}(F)\\\label{pf2_72b}
   =&~\inf_{\substack{\{\mathcal P_t\}_{t\geq0}\in\Pi\colon\\\mathcal P_t=\mathcal P_t^{\mathrm{(a)}},t\leq \tau_i,\\ \eqref{comm_cons_f_b}}}\limsup_{T\rightarrow\infty}\frac{1}{T}\mathbb E\left[\int_{\tau_i}^T(\tilde X_t-\bar {\tilde{X}}_t)^2 dt\right]\\ \label{pf2_72c}
   =&~\inf_{\substack{\{\mathcal P_t\}_{t\geq0}\in\Pi\colon\\ \eqref{comm_cons_f_b}}}\limsup_{T\rightarrow\infty}\frac{1}{T}\mathbb E\left[\int_{0}^{T-\tau_i}(\tilde X_t-\bar {\tilde{X}}_t)^2 dt\right]\\\label{pf2_72d}
   =&~ \underline{D}^{\infty}(F),
\end{align}
\end{subequations}
where \eqref{pf2_72b} is due to assumption (iv); \eqref{pf2_72c} is due to \eqref{pf2_c1}; the equality in \eqref{pf2_72d} is achieved since \eqref{pf2_72c} is upper-bounded by \eqref{pf2_72d} and is equal to \eqref{pf2_72a} simultaneously. Suppose that the sampling-decision processes that achieve \eqref{pf2_72b}--\eqref{pf2_72c} are  $\{\mathcal P_t^{\mathrm{(b)}}\}_{t\geq 0}$ and $\{\mathcal P_t^{\mathrm{(c)}}\}_{t\geq 0}$, respectively. 
From \eqref{pf2_72a} and \eqref{pf2_72b}, we observe that
\begin{equation}\label{pf2_3}
    \left\{\mathcal P_t^{\mathrm{(a)}}\right\}_{t\geq \tau_i}=\left\{\mathcal P_t^{\mathrm{(b)}}\right\}_{t\geq \tau_i},~i=0,1,\dots
\end{equation} 
We prove that under sampling-decision processes satisfying assumption (iv), it holds that
\begin{equation}\label{ETfinite}
    \mathbb E\left[\int_{T-\tau_i}^T(\tilde X_t-\bar {\tilde{X}}_t)^2\right]<\infty,
\end{equation}
so that using \eqref{pf2_72c}, \eqref{pf2_72d}, and \eqref{ETfinite}, we conclude
\begin{equation}\label{pf2_1}
    \left\{\mathcal P_t^{\mathrm{(c)}}\right\}_{t\geq 0}=\left\{\mathcal P_t^{\mathrm{(a)}}\right\}_{t\geq 0}.
\end{equation}
By assumption (iv) we know that there exist sampling-decision processes that lead to
\begin{align}
    \mathbb E\left[\int_{0}^{\tau_i}(\tilde X_t-\bar {\tilde{X}}_t)^2dt\right]<\infty.
\end{align}
Thus, there exist sampling-decision processes such that \eqref{ETfinite} holds. Since the goal is to minimize the MSE, it suffices to consider sampling-decision processes that lead to \eqref{ETfinite}.

Due to \eqref{pf2_c1}, the probability distributions of $\tilde X_t, t\in[0,T-\tau_i]$ in \eqref{pf2_72b} and $\tilde X_t,t\in[\tau_i,T]$ \eqref{pf2_72c} are the same. Thus, the sampling-decision process $\{\mathcal P_t\}_{t\geq \tau_i}=\left\{\mathcal P_{t-\tau_i}^{\mathrm{(a)}}\right\}_{t-\tau_i\geq 0}$ achieves the infimum in \eqref{pf2_72b}. We conclude
\begin{equation}\label{pf2_2}
    \left\{\mathcal P_t^{\mathrm{(b)}}\right\}_{t\geq \tau_i}=\left\{\mathcal P_{t-\tau_i}^{\mathrm{(a)}}\right\}_{t-\tau_i\geq 0}, ~i=0,1,\dots
\end{equation}

Using \eqref{pf2_3} and \eqref{pf2_2}, we conclude that $ \left\{\mathcal P_{t-\tau_i}^{\mathrm{(a)}}\right\}_{t-\tau_i\geq 0}= \left\{\mathcal P_t^{\mathrm{(a)}}\right\}_{t\geq \tau_i}$, $i=0,1,\dots$, i.e.,
\begin{equation}
    a_0(s,0)=a_i(s+\tau_i,\tau_i).
\end{equation}
Thus, \eqref{stp} follows.

Next, we show Remark~\ref{rmk1} using \eqref{stp}. We conclude that the sampling intervals $T_i\triangleq \tau_{i+1}-\tau_i$, $i=0,1,\dots$, are independent due to~\eqref{pf2_c2} and the fact that the sampling-decision process \eqref{stp} is independent of the process prior to the last stopping time; the sampling intervals $T_i$, $i=0,1,\dots$, are identically distributed due to \eqref{pf2_c1} and the fact that the sampling-decision process~\eqref{stp} only takes into account the time elapsed from the last sampling time $t-\tau_i$, $t\in[\tau_i,\tau_{i+1})$, $i=0,1,\dots$

We proceed to show that the optimization problem associated with $\underline D^{\infty}(F)$ can be reduced to \eqref{solvet} by Lemma~\ref{l6}. The assumptions in Lemma~\ref{l6} are satisfied with $Z_i\leftarrow T_i$, $R_i\leftarrow \int_{\tau_i}^{\tau_{i+1}}(X_t-\mathbb E[X_t|X_{\tau_i},\tau_i])^2dt$. The sampling intervals $T_0,T_1,\dots$ are i.i.d. due to Remark~\ref{rmk1}. The expectation of $T_i$ is finite by assumption (iv).
The reward random variables $R_i$ are i.i.d. due to \eqref{pf2_c1}--\eqref{pf2_c2} and Remark~\ref{rmk1}. Furthermore, the expectation of the reward is finite by assumption (iv).
Therefore, using \eqref{pf2_75}, we simplify the DFF in \eqref{dsf_eq} to \eqref{solvet}.

\section{Proof of Theorem~\ref{thm1}}\label{pfthm1}
Achievability is shown right after the statement of Theorem~\ref{thm1}. We here show the converse \eqref{converse}. Denote by $\Pi_{T}$ the set of all sampling-decision processes \eqref{gs'} that satisfy (iv)--(vi) on $[0,T]$. Denote by $\mathcal C_T$ the set of all causal compressing policies on $[0,T]$. We lower bound the DRF in \eqref{DOP} as \eqref{ob} (see \eqref{ob} on the next page),
\begin{figure*}[!t]
% ensure that we have normalsize text
\normalsize
% Store the current equation number.
%\setcounter{MYtempeqncnt}{\value{equation}}
% Set the equation number to one less than the one
% desired for the first equation here.
% The value here will have to changed if equations
% are added or removed prior to the place these
% equations are referenced in the main text.
%\setcounter{equation}{5}
\begin{subequations}\label{ob}
\allowdisplaybreaks
\begin{align}\label{ob_b}
    D(R)
=&~\inf_{\substack{\{\mathcal P_t\}_{t=0}^T\in\Pi_T,\\\{f_t\}_{t=0}^T\in\mathcal C_T\colon\\ \eqref{comm_cons_a}}}\frac{1}{T}\mathbb E\left[\sum_{i=0}^{N}\int_{\tau_i}^{\tau_{i+1}}(X_t-\mathbb E[X_t|U^i,\tau^i,t<\tau_{i+1}])^2dt\right]\\ \label{ob_c}
\geq&~\inf_{\substack{\{\mathcal P_t\}_{t=0}^T\in\Pi_T\colon\\ \frac{\mathbb E[N]}{T}\leq R}}\frac{1}{T}\mathbb E\Biggl[\sum_{i=0}^{N}\int_{\tau_i}^{\tau_{i+1}}(X_t-\mathbb E[X_t|\{X_s\}_{s=0}^{\tau_i},\tau^i,t<\tau_{i+1}])^2dt\Biggr]\\ \label{ob_d}
=& \inf_{\substack{\{\mathcal P_t\}_{t=0}^T\in\Pi_T\colon\\ \frac{\mathbb E[N]}{T}\leq R}}\frac{1}{T}\mathbb E\Biggl[\sum_{i=0}^{N}\int_{\tau_i}^{\tau_{i+1}}(\tilde X_t-\mathbb E[\tilde X_t|\tau_i,t<\tau_{i+1}])^2dt\Biggr]\\
\label{ob_e}
=&~\underline{D}(R).
\end{align}
\end{subequations}
% Restore the current equation number.
%\setcounter{equation}{\value{MYtempeqncnt}}
% IEEE uses as a separator
\hrulefill
% The spacer can be tweaked to stop underfull vboxes.
\vspace*{4pt}
\end{figure*}
% Store the current equation numbe
where \eqref{ob_c} holds since $\mathbb E[N] \leq \mathbb E\left[\sum_{i=1}^{N}\ell(U_i)\right]$, and $U^i$ belongs to the $\sigma$-algebra generated by the stochastic process $\{X_s\}_{s=0}^{\tau_i}$. The equality in~\eqref{ob_d} is obtained by subtracting and adding $\mathbb E[X_t|X_{\tau_i},\tau_i]$ to $X_t$ in \eqref{ob_c}, where
\begin{align}\nonumber
 &\mathbb E[\tilde X_t|\tau_i,t<\tau_{i+1}]\\ =~& \mathbb E[X_t|\{X_s\}_{s=0}^{\tau_i},\tau^i,t<\tau_{i+1}] - \mathbb E[X_t|X_{\tau_i},\tau_i]
\end{align}
holds due to the argument that justifies \eqref{XG'c} with $\{X_{\tau_j}\}_{j=1}^i$ replaced by $\{X_s\}_{s=0}^{\tau_i}$. 

While \eqref{ob} shows that the converse \eqref{converse} holds for the finite horizon ($T<\infty$), the converse also holds for the infinite horizon ($T=\infty$). This is because \eqref{ob} continues to hold with the minimization constraints $\{\mathcal P_t\}_{t=0}^T\in\Pi_T$, $\{f_t\}_{t=0}^T\in\mathcal C_T$, \eqref{comm_cons_a}, and $\frac{\mathbb E[N]}{T}\leq R$ replaced by $\{\mathcal P_t\}_{t\geq 0}\in\Pi_{\infty}$, $\{f_t\}_{t\geq 0}\in\mathcal C_{\infty}$, \eqref{comm_cons_b}, and $\limsup_{T\rightarrow\infty}\frac{\mathbb E[N]}{T}\leq R$, respectively, and with $\limsup_{T\rightarrow\infty}$ inserted right before the objective functions in \eqref{ob_b}--\eqref{ob_d}.

\section{Proof of Proposition~\ref{prop4}}\label{pfprop4}
Using \eqref{cor17}, we calculate that for $t\in[\tau_i,\tau_{i+1})$,
\begin{equation}\label{X_diff}
    X_t-\bar X_t = O_{t-\tau_i} \triangleq \frac{\sigma}{\sqrt{2\theta}}e^{-\theta(t-\tau_i)}W_{e^{2\theta(t-\tau_i)}-1}.
\end{equation}
Let
\begin{equation}\label{sint}
    T_i\triangleq \tau_{i+1}-\tau_i, i=0,1,2,\dots
\end{equation}

We write the objective function of \eqref{solvet} as
\begin{subequations}
\allowdisplaybreaks
\begin{align} \label{oueq1c}
    &\frac{\mathbb E\left[\int_{0}^{T_i}O_t^2dt\right]}{\mathbb E[T_i]}\\ \label{oueq1d}
    =~& \frac{\mathbb E\left[R_2(O_{T_i}^2)\right]}{\mathbb E[R_1(O_{T_i}^2)]}\\ \label{oueq1e}
    \geq ~&
    \frac{R_2(\mathbb E[O_{T_i}^2)])}{ R_1(\mathbb E[O_{T_i}^2])},
\end{align}
\end{subequations}
where \eqref{oueq1c} is obtained by plugging \eqref{X_diff}--\eqref{sint} into \eqref{solvet}; \eqref{oueq1d} holds by solving Dynkin's formula \cite[Eq.(44)]{Zaman} for $R_1(O_{T_i}^2)$ and $R_2(O_{T_i}^2)$ in \eqref{27a}--\eqref{27b} to obtain
\begin{subequations}\label{Dynkins}
\begin{align}
    &\mathbb E\left[\int_{0}^{T_i}O_t^2dt\right]=\mathbb E\left[R_2(O_{T_i}^2)\right],\\ \label{86b}
    &\mathbb E[T_i] = \mathbb E[R_1(O_{T_i}^2)],
\end{align}
\end{subequations}
and plugging \eqref{Dynkins} into \eqref{oueq1c}; \eqref{oueq1e} is obtained by letting 
\begin{subequations}
\begin{align}
    &A\triangleq \frac{\sigma^2}{2\theta}\mathbb E[R_1(O_{T_i}^2)],\\
    &B\triangleq\frac{\sigma^2}{2\theta}R_1(\mathbb E[O_{T_i}^2]),\\
    &C\triangleq\frac{1}{2\theta}\mathbb E[O_{T_i}^2],\\
    &\mathbb E\left[R_2(O_{T_i}^2)\right] = A-C,\\
    &R_2(\mathbb E[O_{T_i}^2)]) = B-C,
\end{align}
\end{subequations}
applying Jensen's inequality to $R_1(v)$, which is convex in $v>0$, to obtain $A \geq B$, and using the fact that $\frac{A-C}{A}\geq \frac{B-C}{B}$ for $A\geq B\geq C\geq 0$.

By \eqref{86b} (\cite[Eq.(43)]{Zaman}) and Jensen's inequality, we write the minimization constraint in \eqref{solvet} as,
\begin{equation}\label{range}
   R_1(\mathbb E[O_{T_i}^2]) \leq \mathbb E[R_1(O_{T_i}^2)] = \mathbb E[T_i]= \frac{1}{R}.
\end{equation}
For any $R_1(\mathbb E[O_{T_i}^2])$ in the range \eqref{range}, \eqref{oueq1e} is a lower bound to \eqref{oueq1c}. Choosing $R_1(\mathbb E[O_{T_i}^2])$ that satisfies \eqref{range} with equality leads to \eqref{oueq1e} being equal to $D(R)$ in \eqref{dr}. 

Plugging \eqref{opt_ou} into \eqref{oueq1d}, we verify that the lower bound in \eqref{oueq1d} is achieved by the symmetric threshold sampling policy in \eqref{opt_ou}.

\section{Recovering $L_t$ from $Z_t$}\label{pf_LZ}
Denote by $\delta(\cdot)$ the Dirac-delta function. Let $\nu_i$ be the $i$-th discontinuous point of $\{Z_t\}_{t=0}^T$. For $\{Z_t\}_{t=0}^T$ in \eqref{ZX}, $\nu_i$ is simply equal to the sampling times $\tau_i$, $i=1,2,\dots$ Without loss of generality, we assume that $\{Z_{t}\}_{t=0}^T$ is right-continuous at the discontinuous point $\nu_i$, since the mean-square cost in~\eqref{plant_MSE} is not affected by the assumption.  Denote by $\nu_i^-$ the time just before time $\nu_i$, where $Z_{\nu_i^-}\neq Z_{\nu_i}$.

\begin{proposition}\label{prop_lz}
Assume that $\{Z_t\}_{t=0}^T$ is almost surely $ACG_*$ on $[\nu_i,\nu_{i+1})$ and is right-continuous at the discontinuous point $\nu_i$. Then, control signal $\{L_t\}_{t=0}^T$ in \eqref{SDE_plant}  for $t\in[\nu_i,\nu_{i+1})$, $i=1,2,\dots$, is given by
\begin{align}\label{lt}
    L_t = \begin{cases}
      \left(Z_{\nu_i} - Z_{\nu_i^-}\right)\delta(t-\nu_i), & t=\nu_i, \\
      \lim_{\delta\rightarrow 0^+} \frac{Z_t-Z_{t-\delta}}{\delta}, &t\in(\nu_i,\nu_{i+1}).
    \end{cases}
\end{align}
\end{proposition}
\begin{proof}
For $t\in[\nu_i,\nu_{i+1})$, we rewrite \eqref{LZ} as
\begin{subequations}\label{LZ1}
\begin{align}
    \int_{\nu_i}^t L_s \mathrm{d}s & = Z_t -\lim_{\delta\rightarrow 0^+} \int_{0}^{\nu_i-\delta}L_s\mathrm{d}s\\\label{LZ1b}
    & =Z_t - Z_{\nu_i^-}\\\label{LZ1c}
    & = (Z_t - Z_{\nu_i}) + (Z_{\nu_i} - Z_{\nu_i^-}),
\end{align}
\end{subequations}
which is equivalent to \eqref{lt}. 
\end{proof}
Note that $L_{\nu_i}$ is an impulse control at $t=\nu_i$ \cite{1}, \cite{Lions}--\cite{ZGLi}, and $L_t$, $t\in(\nu_i,\nu_{i+1})$ is equal to the left-derivative of $Z_t$. This is because $Z_t$ may not be differentiable at $t$, but its left-derivative exists since the $ACG_*$ property of $Z_t$ implies that it is differentiable almost everywhere on $(\nu_i,\nu_{i+1})$ \cite{Talvila}.
For example, if $X_t = W_t$, the optimal control signal \eqref{ZX} is $Z_t = -W_{\tau_i}$, $t\in[\tau_i,\tau_{i+1})$, and the corresponding control signal in~\eqref{SDE_plant} is $L_{t} = -(W_{\tau_i}-W_{\tau_{i-1}})\delta(t-\tau_i)$ for $t\in [\tau_i,\tau_{i+1})$.

\ifCLASSOPTIONcaptionsoff
  \newpage
\fi



%\bibliographystyle{IEEEtran}


\begin{IEEEbiographynophoto}
{Nian Guo} is a Ph.D. candidate in Electrical Engineering at California Institute of Technology. She received the Bachelor's degree in Electrical Engineering from The University of Hong Kong in 2017 and the Master's degree in Electrical Engineering from California Institute of Technology in 2019. Her research interests include information theory, theory of random processes, control, and wireless communication.
\end{IEEEbiographynophoto}

\begin{IEEEbiographynophoto}
{Victoria Kostina}(S'12--M'14)
 received the bachelor's degree from Moscow Institute of Physics and Technology (MIPT) in 2004, the master's degree from University of Ottawa in 2006, and the Ph.D. degree from Princeton University in 2013.  During her studies at MIPT, she was affiliated with the Institute for Information Transmission Problems of the Russian Academy of Sciences. 
She is currently a Professor of electrical engineering and computing and mathematical sciences at California Institute of Technology. Her research interests include information theory, coding, control, learning, and communications. 
 She received the Natural Sciences and Engineering Research Council of Canada postgraduate scholarship during 2009--2012, Princeton Electrical Engineering Best Dissertation Award in 2013, Simons-Berkeley research fellowship in 2015 and the NSF CAREER award in 2017.  
\end{IEEEbiographynophoto}

%\newpage


%\vfill


\end{document}